\documentclass[12pt]{article}
\usepackage[margin=1in]{geometry}
\usepackage{amsmath,amssymb,amsthm,bbm}
\usepackage{mathtools}
\usepackage{enumerate}
\usepackage{algorithm,algpseudocode}
\usepackage{xcolor,soul}
\usepackage{authblk}
\usepackage[makeroom]{cancel}
\usepackage{cases}
\usepackage{booktabs}
\usepackage{graphicx}

\usepackage[numbers]{natbib}

\usepackage{hyperref}
\usepackage{cleveref}

\hypersetup{
    colorlinks,
    linkcolor={red!50!black},
    citecolor={blue!50!black},
    urlcolor={blue!80!black}
}

\algnewcommand{\Initialize}[1]{%
  \State \textbf{Initialize:}
  \Statex \hspace*{\algorithmicindent}\parbox[t]{.96\linewidth}{\raggedright #1}
}

\newtheorem{theorem}{Theorem}
\newtheorem{assumption}{Assumption}
\newtheorem{corollary}{Corollary}
\newtheorem{lemma}{Lemma}
\newtheorem{proposition}{Proposition}

\newtheorem{claim}{Claim}
\newtheorem{definition}{Definition}

\begin{document}
\title{Scheduling with Predictions}


\author{Woo-Hyung Cho\thanks{wc563@cornell.edu}}
\author{Shane Henderson\thanks{sgh9@cornell.edu}}
\author{David Shmoys\thanks{david.shmoys@cornell.edu}}

\affil{Cornell University}

\maketitle

\begin{abstract}
	There is significant interest in deploying machine learning algorithms for diagnostic radiology, as modern learning techniques have made it possible to detect abnormalities in medical images within minutes. While machine-assisted diagnoses cannot yet reliably replace human reviews of images by a radiologist, they could inform prioritization rules for determining the order by which to review patient cases so that patients with time-sensitive conditions could benefit from early intervention.

	We study this scenario by formulating it as a learning-augmented online scheduling problem. We are given information about each arriving patient’s urgency level in advance, but these predictions are inevitably error-prone. In this formulation, we face the challenges of decision making under imperfect information, and of responding dynamically to prediction error as we observe better data in real-time. We propose a simple online policy and show that this policy is in fact the best possible in certain stylized settings. We also demonstrate that our policy achieves the two desiderata of online algorithms with predictions: consistency (performance improvement with prediction accuracy) and robustness (protection against the worst case). We complement our theoretical findings with empirical evaluations of the policy under settings that more accurately reflect clinical scenarios in the real world.
\end{abstract}

\newpage
\section{Introduction}
Modern machine learning algorithms have been tremendously successful in a variety of application domains, and healthcare is no exception. In recent years, we have seen significant interest in deploying these algorithms for diagnostic radiology, a branch of medicine that uses imaging techniques such as X-rays, ultrasounds, and Magnetic Resonance Imaging (MRI) to diagnose a patient. The idea is to use these images as inputs to machine learning algorithms, which would then search for patterns that imply the presence of an abnormality. Advances in pattern recognition techniques for image processing and computer vision have made it possible for machine learning algorithms to detect abnormal conditions in medical images within minutes, or even seconds. Because this is still a nascent area of research, these algorithmic, machine-assisted diagnoses cannot yet reliably replace the thorough, human reviews of images by a radiologist. Meanwhile, they could be used to prioritize and speed up the review of images that are flagged as likely to contain time-sensitive conditions.

To make this more concrete, imagine a group of patients who have had diagnostic images taken after a referral. Radiologists are tasked with processing each patient case, which typically consists of reading images, then communicating any findings by filing a radiology report and sending it back to the referring provider. Appropriate patient care and treatment begin only upon case completion at the radiology department, so it is in the best interest of the patient for radiologists to organize their workflow in a way that prioritizes cases by urgency. This is especially true for patients with time-sensitive conditions such as stroke, intercranial hemorrhage, or pneumothorax, for which early intervention is key. In the case of acute stroke due to large vessel occlusion, for example, studies have shown that an interventional radiology procedure called mechanical thrombectomy could achieve a favorable clinical outcome when performed within 4 to 6 hours of symptom onset \cite{thrombectomy}. This is where machine learning could be helpful. Leveraging the speed and the predictive power of machine learning, radiologists could use algorithmic outputs to prioritize cases that are deemed urgent.

True urgency, however, cannot be fully assessed until a case is opened and images are at least partially read. For this reason, many imaging clinics including those in the New York-Presbyterian hospital network tend to rely on the referring providers' communication of expectations as well as on their own insights, expertise and experience when prioritizing cases. In some sense, current practices rely on \textit{human predictions} of urgency. The use of predictions powered by machine learning algorithms could augment current best practices and streamline the process of determining the order by which patient cases should be read.

But predictions, human-made or machine-learned, are rarely perfect. There will always exist never-before-seen cases that further compound the error. Good predictions have the potential to expedite the detection and treatment of time-sensitive conditions, but mispredictions could cause delays that are extremely costly. Given this understanding, the central question that we ask in this paper is, \textit{how can we take advantage of predictions to improve radiologists' workflow while accounting for prediction error?}

We abstract the setting described above and model it as a single-machine scheduling problem. A radiologist tasked with reviewing patient cases can be viewed as a single machine that is able to process one job at a time. Case urgencies are captured in the form of job weights, where the higher the weight, the greater the urgency. Patient treatment plans are often established upon completion of case review from the radiology department, so a natural objective would be to minimize the total sum of urgency-weighted completion times across all patients.

This single-machine problem of minimizing the weighted sum of job completion times is a decades-old problem that has already been extensively studied (see \cite{smith56,conwaymm67}, for example). In this paper, we study this problem with the addition of a key feature: imperfect predictions of urgency. Patient cases randomly arrive into the system. At each job's time of arrival, we observe its predicted level of urgency given by some black-box predictive mechanism. True urgency is unknown and unobservable at this time, so priority decisions are necessarily made based on imperfect information. However, when radiologists are working through a patient case, interpreting the associated images and deciding if an abnormality is present, they are also gradually learning whether or not the case on hand is truly urgent (or non-urgent). Anecdotal evidence suggests that an image study is roughly a process of elimination via inspection from different angles \cite{radiologist_interpretations2,radiologist_interpretations}, so it is likely that a job's true urgency is known even before its processing is complete. We therefore allow radiologists to \textit{preempt} a job midway then return to the remaining work for completion at a later point in time. With a preemptive strategy, we have an opportunity to hedge against prediction error by responding early to what is in hindsight a suboptimal decision made in the face of less-than-perfect information. We aim to find a policy for deciding which job or remainder thereof to process at any given time so that the total expected urgency-weighted sum of job completion times is minimized.

Our scheduling formulation allows for a wide range of models for describing the problem setting in ways that more accurately reflect clinical settings in the real world. For example, by using job weights to capture case urgencies, we are able to handle granularity in prioritization schemes beyond a binary classification of urgent vs. non-urgent. Our preemptive framework also allows flexibility in modeling the many different ways in which radiologists gain information as they process each patient case. Nevertheless, in this paper, we focus on a highly stylized version of this model. We assume that each job can be categorized as one of two types: urgent or non-urgent. All jobs are available before any decisions are made, and we are able to observe each job's predicted priority class at this time. We further assume that every job shares the same processing time requirement. Without loss of generality, we assume unit processing time requirements. A fixed parameter $\alpha\in(0,1)$ is used to denote the fraction of a job that must be processed before we learn its true type. We call this time point a job's $\alpha$-point. In our model, we allow preemptions to occur only at these $\alpha$-points, and assume that the residual work needed to complete an interrupted job is exactly the same as if the job had not been interrupted.

Our problem of scheduling with predictions is an exercise in online decision making even when all jobs are available to us in advance. Decisions are made with incomplete information in the form of imperfect predictions, to which we respond over time based on our observations of true job types. Classic results in online decision making have focused on finding solutions that are robust with provably good performance guarantees over all possible inputs and even in the worst case. An emerging line of research in this area leverages predictions to design algorithms that not only remain robust to worst-case inputs but also achieve performance guarantees that improve with prediction accuracy (see \cite{mitzenmacherv20} for a survey). We continue this line of research and extend it to our problem setting. In what follows, we first find a threshold-based policy for deciding which job to process at any given time, and show that our proposed policy is the best possible over all non-anticipating policies. We then show that performance guarantees for this policy degrade gracefully as a function of prediction error. Our results indicate that our policy simultaneously achieves consistency (improvement with prediction accuracy) and robustness (protection against the worst case).

\paragraph{Related Work}

There has been an explosion of research activity in recent years that seeks to augment online algorithms with machine-learned predictions. In this framework, the goal is to design algorithms with near-optimal performance when predictions are accurate while maintaining prediction-less guarantees in the worst case. The idea is that good predictions can help circumvent worst-case behavior. Classic optimization problems that are being reexamined under this framework include caching \cite{lykouris21}, matching \cite{dinitzilm21,antoniadis20,chensvz22,lavastidamrx21}, secretary \cite{duttingllv21}, knapsack \cite{imkmqp21} and facility location \cite{almanzaclpr21}. Problems in Nash social welfare \cite{banerjeeggj22}, mechanism design \cite{xulu22} and revenue management \cite{balseirokk22} are also actively being studied in this context.

In online scheduling, problems that are being newly examined with learning augmentation include problems for minimizing average flow time \cite{mitzenmacher20,mitzenmacherv20,azarlt21}, average completion time \cite{purohitsk18, im21}, average weighted completion time \cite{megow22}, and makespan \cite{lattanzi20,balkanski22}. Many of these studies with min-sum objectives assume that job processing requirements are not known to us in advance. In these settings, it is natural to use predictions of individual job processing times \cite{purohitsk18, im21, mitzenmacher20, mitzenmacherv20, azarlt21}. More recent work examines the use of permutation predictions, directly predicting algorithmic actions rather than input characteristics \cite{megow22}. Our work studies the min-sum weighted completion time objective using predictions of an input characteristic that has not been considered in previous work: job weights.

Job weights are used to capture urgencies or priorities in our problem setting. Outside the realm of online scheduling with predictions, there is an extensive body of work that investigates the effect of priority classes. In the context of prediction error, our work is closely related to the works of \citet{argon09} and \citet{mclaym13}. In a priority queue model, \citet{argon09} make priority assignments for arriving customers based on imperfect indicators of priority types. The signal available to the decision maker is the probability that a customer is high priority. \citet{mclaym13} study the problem of dispatching ambulances when operators make classification errors in assessing patient risk via a Markov Decision Process. Our model is fundamentally different not just in framework, but more importantly in that we respond dynamically to real-time information gained while processing each job. Despite these major modeling differences, there are striking similarities in some of the insights and conclusions we draw. With \citet{argon09}, we share the same optimal policy structure given two priority classes with linear waiting costs. Both works have a signal (or, in our case, prediction)-based thresholding policy with strong ties to the generalized $c\mu$ rule. Our threshold policy also reveals how prediction quality impacts decision making. Similar insights are given by \citet{mclaym13} on when to over- or under-respond to perceived patient risk based on rates of classification error. It is clear that there are connections in our approaches despite their differences. In future work, it would be interesting to see when and how these frameworks converge.

In other related work, \citet{vanderzee61} directly model misclassification rates in a single-server queue where priority assignments are made based on a probabilistic classifier. Steady-state results are derived when classification errors are known and very small. \citet{singhgv20} eliminates the use of priority types as a middle-man altogether and directly prescribes placement into the priority queue. Finally, very recent work by \citet{thompson2022} explores the impact of prediction-driven prioritization schemes using a preemptive priority queue. Their simulated clinical impact assumes a fixed prediction error based on the expected diagnostic performance of machine-learned algorithms.

The remainder of this paper is organized as follows. In Section \ref{section:formulation}, we introduce our model as well as a scheduling formulation of the problem. We present our main results in Section \ref{section:beta}, where we show that a simple threshold-based policy is in fact the best possible in certain stylized settings. Section \ref{sectionch2:extensions} extends this idea to a number of settings that more accurately reflect realistic scenarios. Finally, we conclude and lay out some additional thoughts for future research in Section \ref{sectionch2:discussion}.


\section{Problem Formulation}\label{section:formulation}

We have a set of patient cases that must be processed by a radiologist. At time of arrival, each patient case is labeled with its predicted urgency level. These labels are observable. At every decision point, the radiologist decides which patient case to process. After processing a pre-specified fraction of a patient case, the radiologist learns the true priority of the case on hand and has the option to preempt that case in favor of another patient case. We capture this decision making process with a preemptive scheduling model. Our goal is to find a policy for minimizing the expected urgency-weighted sum of completion times across all patients. We describe the problem data and model, followed by a scheduling formulation of the problem.

\paragraph{Problem Data}
We have one radiologist (a single machine) processing patient cases (jobs) indexed by $[n]=\{1,\dots,n\}$. A machine can only process one job at a time, and each job requires 1 unit in processing time. All jobs are assumed available at time 0 in advance of any decision making, i.e., release dates $r_j=0$ for each job $j\in [n]$.

There are two priority classes, type 0 (urgent) and type 1 (non-urgent), each with its associated cost per unit delay (weights) $\omega_0$ and $\omega_1$, respectively, where $\omega_0> \omega_1>0$. Each job is independently an urgent job with probability $\rho\in(0,1)$, which we assume is known based on historical data. Job $j$'s true urgency $true(j)\in\{0,1\}$ is unknown a priori, and is revealed only after partially completing some fixed $\alpha\in (0,1)$ fraction of the job. On the other hand, its predicted priority $pred(j)\in\{0,1\}$ is immediately observable at its release date $r_j$. A binary classification system predicts the urgency level of each job independently according to the following probability matrix.

\begin{table}[ht]
	\centering
	\begin{tabular}{|c|c|c|}\hline
	&predicted 0& predicted 1\\
	\hline
	true 0 & $1-\varepsilon_0$&$\varepsilon_0$\\
	\hline
	true 1 & $\varepsilon_1$&$1-\varepsilon_1$\\
	\hline
	\end{tabular}
	\caption{Prediction probability matrix}
	\label{tab:pred-matrix}
\end{table}

The probability of misclassifying a true type 0 job is the false negative rate $\varepsilon_0$, and the probability of misclassifying a true type 1 job is the false positive rate $\varepsilon_1$. We assume that $\varepsilon_0\leq 1/2$ and $\varepsilon_1\leq 1/2$, and that these prediction errors are known. We expect that they could be inferred from historical data or from expected generalization error rates associated with the machine learning algorithm that we use.

By Bayes' rule, job $j$ is a type 0 job with probability $p_j$, where 
	\begin{equation}
		p_j=\mathbb{P}(true(j)=0|pred(j)=0)=\frac{(1-\varepsilon_0)\rho}{(1-\varepsilon_0)\rho+\varepsilon_1(1-\rho)}\label{eq:pj-1}
	\end{equation}
	if job $j$ is predicted to be of high priority, and 
	\begin{equation}
		p_j=\mathbb{P}(true(j)=0|pred(j)=1)=\frac{\varepsilon_0\rho}{\varepsilon_0\rho+(1-\varepsilon_1)(1-\rho)}\label{eq:pj-2}
	\end{equation} otherwise. It is easy to verify that $\mathbb{P}(true(j)=0|pred(j)=0)\geq \mathbb{P}(true(j)=0|pred(j)=1)$ given our assumptions that $\varepsilon_0$ and $\varepsilon_1$ are both at most one half. Finally, the weight of job $j$ is
\begin{align}
	w_j &= \omega_0\cdot\mathbf{1}\left\{true(j)=0\right\} + \omega_1\cdot\mathbf{1}\left\{true(j)=1\right\}\nonumber\\
	&=\omega_1+(\omega_0-\omega_1)\cdot\mathbf{1}\left\{true(j)=0\right\}.\label{eq:job-weight}
\end{align}

\begin{assumption}\label{asm:wsrpt}
	$\omega_1<\omega_0(1-\alpha)$ holds.
\end{assumption}
Intuitively, Assumption \ref{asm:wsrpt} ensures that there is a large enough weight differential between urgent and non-urgent jobs to make preemption meaningful. The technical reasons for making this assumption will be discussed in the next section when it becomes relevant.

\paragraph{Model}
A decision point occurs whenever a job completes or a job's true priority is revealed. We call the latter decision point an $\alpha$-point. Our model allows \textit{preemptions}; at each $\alpha$-point, we can either complete the job immediately, or preempt then process the remaining $1-\alpha$ units of work at a later point in time. 

At each decision point $t$, we observe the state, which consists of the set of unopened jobs sorted in some order, and the set of partially processed jobs of which true types are already known. Of unopened jobs, only predicted priorities are known. Based on the state, we decide whether to open a new job of as-yet-unknown urgency or complete a job of known priority that only has $1-\alpha$ units of work remaining. Our decisions at each decision point are therefore made based on the \textit{predicted} priorities of unopened jobs and the \textit{true} priorities of partially processed jobs. If we decide to open a new job, we process a job chosen according to some predetermined order and meet our next decision point at the next $\alpha$-point $t+\alpha$, at which time we observe the job's true type. We then update the state by moving this job from the set of unopened to the set of partially processed jobs. Otherwise, we complete a job at $t+(1-\alpha)$, incur a weighted cost to the objective based on the true urgency of the job just completed, and remove that job from the system entirely.

\paragraph{Objective}
Each of the $n$ arriving jobs is independently a high priority job with probability $\rho$, and is assigned a predictive label according to the probability matrix given in Table \ref{tab:pred-matrix}. Letting $C_j$ denote the completion time of job $j$, our goal is to minimize $\mathbb{E}\left(\sum_{j=1}^n w_jC_j\right)$ where $w_j$ is the true weight of job $j$ as defined in Equation \eqref{eq:job-weight}.

\paragraph{Scheduling Formulation}
We first consider the offline version of this problem in which jobs' true types are known a priori. This is a single-machine problem of minimizing the weighted sum of completion times, written $1||\sum w_jC_j$ in the scheduling notation of \citet{grahamllr79}, and can be solved using the following result given by \citet{smith56}.

\begin{theorem}[Smith's WSPT Rule]
	For the single-machine problem of minimizing the weighted sum of completion times, the Weighted Shortest Processing Time (WSPT) rule is optimal.\label{thm:smith-wspt}
\end{theorem}

The WSPT rule sorts jobs in nonincreasing order of weight-to-processing-time ratios. Given our unit processing time assumption, sorting jobs in WSPT order is equivalent to sorting jobs in nonincreasing order of true priorities. Then, processing job $j$ for completion at time $j$ yields an optimal schedule, so $\mathbb{E}(\mathsf{OPT})= \mathbb{E}\left(\sum_{j=1}^n jw_j\right).$

Our problem, however, is a \textit{non-clairvoyant} online decision making problem in which jobs' true priorities are not known until jobs are at least partially processed. We follow the WSPT rule and sort jobs in nonincreasing order of \textit{predicted} weights, breaking ties arbitrarily. We then proceed by opening jobs in this sorted order. The rest of this paper is focused on showing that the performance gap between the online and offline versions of this problem can be reduced with the use of \textit{predictions}, especially when the predictor has low error.

\paragraph{Example}
Consider the following deterministic 9-job example where each column represents a single job.

\begin{table}[ht]
	\centering
	\begin{tabular}{|c|c c c c c c c c c|}
	\hline
	true types (unknown) &0&0&0&0&1&1&1&1&1\\
	\hline
	predicted types (observed) &0&1&0&0&0&1&1&1&0\\
	\hline
	\end{tabular}
\end{table}

At each job's release date, we observe its predicted type. True types are not known at this time. We proceed by sorting jobs in WSPT order of predicted priorities.

\begin{table}[ht]
	\centering
	\begin{tabular}{|c|c c c c c c c c c|}
	\hline
	predicted types (sorted) &0&0&0&0&0&1&1&1&1\\
	\hline
	true types (permutation $\pi$) &0&1&0&0&1&1&1&0&1\\
	\hline
	\end{tabular}
	\caption{An example of a possible job ordering}
\end{table}

Among the five jobs predicted to be of high priority, there are two jobs that are actually of type 1. Similarly, there is one true type 0 job among the four jobs that are predicted to be of low priority. Therefore, in this example, there are ${5\choose 2}{4\choose 1}$ possible misprediction-driven ways of sequencing jobs by true priority. One such permutation $\pi$ is given as an example above. 

Once a job is opened, we learn its true type after processing $\alpha$ units of the job. At this $\alpha$-point, we have the option of either completing the remaining $1-\alpha$ units of work, or opening the next job in $\pi$. Decisions are made over time with the goal of minimizing $\mathbb{E}\left(\sum_{j=1}^n w_jC_j\right)$ across all possible permutations of job orderings.

\section{The $\beta$-Threshold Rule: An Optimal Policy}\label{section:beta}
We provide an optimal policy for our problem in this section. Before we do so, we first discuss an old scheduling result by \citet{schrage68}.

\begin{theorem}[Schrage's SRPT Rule]
	In the preemptive single-machine problem of minimizing the sum of completion times, where jobs are arriving over time, the Shortest Remaining Processing Time (SRPT) rule is optimal.\label{thm:schrage-srpt}
\end{theorem}

By the SRPT rule, given any two jobs of the same weight, we should process the job with the shorter amount of remaining work first. SRPT applied to our problem confirms a general intuition that it is never optimal to preempt a job that is revealed to be of high priority; type 0 jobs will always be processed nonpreemptively. This does not change our model, but it does help simplify some aspects of it. Since preemption only occurs on type 1 jobs, we are able to eliminate $\alpha$-points with respect to type 0 jobs. It also suffices to track the number of partially processed jobs as these jobs are all of type 1.

We now present a thresholding policy that minimizes our objective across all non-anticipating policies. Without loss of generality, we sort jobs according to their predicted priorities, breaking ties arbitrarily. This is equivalent to sorting jobs in nonincreasing order of their type 0 probabilities $p_j$ as defined in equations (\ref{eq:pj-1})-(\ref{eq:pj-2}). As we proceed with our policy, we open jobs in this order. Define a constant $$\beta=\frac{\alpha}{1-\alpha}\cdot\frac{\omega_1}{\omega_0-\omega_1}.$$ At each decision point, we observe the state $(\mathcal{S},\ell)$, where $\mathcal{S}$ is the set of unopened jobs and $\ell$ is the number of partially processed type 1 jobs. If either $\mathcal{S}=\emptyset$ or $\ell=0$, we do not have a decision to make; if the former, we complete a partially processed type 1 job, and if the latter, we open a new job. If $\mathcal{S}=\emptyset$ and $\ell=0$, we are done. Thus, we assume that $\ell>0$ and $\mathcal{S}\neq\emptyset$ with $k=\min(\mathcal{S})$, which means that job $k$ is the next job in line. We make our decisions by comparing $p_k$ against $\beta$: if $p_k\leq\beta$, we process the remaining $1-\alpha$ units of a partially completed low priority job and reach our next decision point at completion. If $p_k>\beta$, we open job $k$ and process $\alpha$ units of the job, at which time we learn of job $k$'s true type. If job $k$ is a type 1 job, we are at our new decision point. Otherwise, we process job $k$ to completion for another $1-\alpha$ units and make our next decision when job $k$'s processing is complete.

\begin{algorithm}
	\caption{The $\beta$-Threshold Rule}
	\label{alg:beta}

	\begin{algorithmic}[1]
	\Require{jobs sorted in nonincreasing order of $p_j$}
	
	\Initialize{$t\leftarrow 0$ \Comment{time} \\ $\mathcal{S}\leftarrow[n]$ \Comment{set of unopened jobs} \\ $\ell\leftarrow 0$ \Comment{number of partially completed type 1 jobs}}
	
	\Procedure{CompleteLow}{$t,\mathcal{S},\ell$} \Comment{complete a known low priority job}
		\State $\ell\leftarrow\ell-1$
		\State $t\leftarrow t+(1-\alpha)$
	\EndProcedure
	
	\Procedure{OpenNext}{$t,\mathcal{S},\ell$} \Comment{open a new job, then stop at the $\alpha$-point}
		\State $k\leftarrow\min(\mathcal{S})$
		\State $\mathcal{S}\leftarrow\mathcal{S}\setminus\{k\}$
		\State $t\leftarrow t+\alpha$
		\If{$true(k)=0$}\Comment{nonpreemptively complete a type 0 job}
			\State $t\leftarrow t+(1-\alpha)$
		\Else
			\State $\ell\leftarrow\ell+1$
		\EndIf
	\EndProcedure
	
	\While{$(\mathcal{S},\ell)$ is not $(\emptyset,0)$}
		\If{$\mathcal{S}=\emptyset$} 
			\State \textsc{CompleteLow$(t,\mathcal{S},\ell)$}
			
		\ElsIf{$\ell=0$}
			\State \textsc{OpenNext$(t,\mathcal{S},\ell)$}

		\Else
			\State $k\leftarrow\min(\mathcal{S})$
			\If{$p_k>\beta$}
				\State \textsc{OpenNext$(t,\mathcal{S},\ell)$}
			\Else
				\State \textsc{CompleteLow$(t,\mathcal{S},\ell)$}
			\EndIf
		\EndIf
	\EndWhile
	\end{algorithmic}
	
\end{algorithm}

\begin{theorem}\label{thm:beta-optimal}
	The $\beta$-threshold rule is optimal.
\end{theorem}
\begin{proof}
	Suppose on the contrary that there exists an optimal policy that does not follow the $\beta$-threshold rule. By assumption, if we run this policy on any instance of our problem input, there exists at least one decision point where the optimal policy observes the given state and makes a decision that deviates from ours. We consider the \textit{last} such decision point and call it time $t$, and further assume that the observed state at that time is $(\mathcal{S},\ell)$, where $\mathcal{S}$ is the set of unopened jobs such that $k = \min(\mathcal{S})$ and $\ell>0$ is the number of partially processed type 1 jobs.
	
	Two things may have occurred at time $t$: $p_k>\beta$ and the optimal policy processes a low priority job for completion at time $t+(1-\alpha)$, or $p_k\leq\beta$ and the optimal policy proceeds by opening job $k$. In both cases, we show that choosing the alternative improves the objective value and ensures that the resulting schedule is consistent with the $\beta$-threshold rule.
	
	\begin{enumerate}[i]
		\item $p_k>\beta$: according to the $\beta$-threshold rule, we should have opened job $k$ at time $t$; the optimal policy decided otherwise and completed a type 1 job (let us call this job $i$) at time $t+(1-\alpha)$. We proceed by identifying another point in time in the schedule generated by the optimal policy to process job $i$, which would allow job $k$ to be processed at time $t$ instead. We then show by an interchange argument that doing so improves the objective value.

		Starting from time $t$, trace time forward in the schedule generated by the optimal policy. Since $t$ is the last decision point that deviates from the $\beta$-threshold rule by assumption and the set of unopened jobs at the next decision point $t+(1-\alpha)$ remains unchanged so that $k=\min(\mathcal{S})$, the optimal policy opens job $k$ at time $t+(1-\alpha)$. We continue to trace time forward until some time $u$, when the optimal policy begins processing the remaining $1-\alpha$ units of a previously preempted, true low priority job for the first time since completing job $i$ at time $t+(1-\alpha)$. We show that within the interval $[t,u)$, we can improve the objective by delaying the completion of job $i$ to $C_i=u$ and moving up the schedule in $[t+1-\alpha,u)$ by $1-\alpha$ units to $[t,u-1+\alpha)$.
		
		By our assumptions, $u$ is the time at which either $\mathcal{S}=\emptyset$ (every job has been opened), or the next job's type 0 probability falls below $\beta$, whichever happens first. Therefore, all jobs opened in $[t+1-\alpha,u)$ are above the probability threshold $\beta$. Let $z\geq 1$ denote the number of jobs opened in $[t+1-\alpha,u)$ including job $k$. Among these $z$ jobs, suppose there are $z_0$ jobs of true type 0 that are completed immediately where $z_0=\sum_{j=k}^{k+z-1}\mathbf{1}\left\{true(j)=0\right\}$. The remaining $z-z_0$ jobs are revealed to be of true type 1 after $\alpha$ units of processing, then are preempted. These preempted jobs are not processed until at least time $u$.
		
		By interchange, each of the $z_0$ type 0 jobs complete $1-\alpha$ units earlier. On the other hand, completion of job $i$, a type 1 job, is delayed by $u-(t+1-\alpha) = z_0+(z-z_0)\alpha$, which sums one unit of delay for every completed type 0 job and an $\alpha$ unit of delay for every preempted type 1 job. None of the other jobs are affected by this interchange. Thus, the overall change to the objective is
		\begin{eqnarray}
			&&-z_0\omega_0(1-\alpha)+\omega_1\left(z\alpha+z_0(1-\alpha)\right)\label{eq:thm3-change1}\\
			&=&-z_0(\omega_0-\omega_1)(1-\alpha)+\alpha \omega_1z \nonumber\\
			&=& -(\omega_0-\omega_1)(1-\alpha)\left(z_0-\frac{\alpha}{1-\alpha}\cdot\frac{\omega_1}{\omega_0-\omega_1} z\right)\nonumber\\
			&=&-(\omega_0-\omega_1)(1-\alpha)\left(z_0-\beta z\right)\nonumber\\
			&=&-(\omega_0-\omega_1)(1-\alpha)\sum_{j=k}^{k+z-1}\left(\mathbf{1}\left\{true(j)=0\right\}-\beta\right).\nonumber
		\end{eqnarray}
		
		In expectation, $$-(\omega_0-\omega_1)(1-\alpha) \sum_{j=k}^{k+z-1}\left(p_j-\beta\right)<0$$ since $p_j>\beta$ for each $j=k,\dots,k+z-1$, which establishes a contradiction. We have also shown how to choose the interval $[t,u)$ for interchange to ensure that the schedule is consistent with the $\beta$-threshold rule from time $t$ onward.
				
		\item $p_k\leq\beta$: according to the $\beta$-threshold rule, we should have completed a type 1 job for completion at time $t+1-\alpha$. Instead, the optimal policy opened job $k$. The proof proceeds similarly to the above in that we first identify an appropriate point in time in the schedule generated by the optimal policy to open job $k$, then use interchange. The main difference lies in that the amount of delay from opening job $k$ is not immediately clear, since that depends on job $k$'s true type.
		
		We first argue that in fact, regardless of type, job $k$ completes at $t+1$ in the schedule generated by the optimal policy. This is trivially true if job $k$ is a type 0 job. Otherwise, the optimal policy meets its next decision point at $t+\alpha$, where the set of unopened jobs is $\mathcal{S} = [n]\setminus [k]$ and there are now $\ell+1$ true type 1 jobs that are not yet fully processed including job $k$. By our assumption that time $t$ is the last decision point at which the optimal policy deviates from the $\beta$-threshold rule, the optimal policy completes a type 1 job at $(t+\alpha)+(1-\alpha) = t+1$ since $\beta\geq p_k\geq p_{k+1}$ at time $t+\alpha$. We are free to label this job as job $k$. Then, starting from $t+1$, 	the optimal policy will complete the remaining $\ell$ type 1 jobs in succession, completing the last type 1 job at time $t+\ell(1-\alpha)+1$.
		
		We show that within the interval $[t,t+\ell(1-\alpha)+1)$, we can improve the objective by delaying the opening of job $k$ to time $t+\ell(1-\alpha)$, when $\ell$ type 1 jobs have each completed the remaining $1-\alpha$ units of work. Even with this interchange, since $\beta\geq p_{k+1}$, job $k$ will be processed nonpreemptively regardless of type so that $C_k=t+\ell(1-\alpha)+1$. None of the other jobs are affected. The overall change to the objective is
		\begin{eqnarray}
			&&-\ell \omega_1  + w_k\ell(1-\alpha)\label{eq:thm3-change2}\\
			&=& -\ell \omega_1  + \left(\omega_1+(\omega_0-\omega_1)\cdot\mathbf{1}\left\{true(k)=0\right\}\right)\ell(1-\alpha)\qquad\qquad\qquad\ \text{by \eqref{eq:job-weight}}\nonumber\\
			&=& -\ell \omega_1  +\omega_1\ell(1-\alpha)+(\omega_0-\omega_1)\mathbf{1}\left\{true(k)=0\right\}\ell(1-\alpha)\nonumber\\
			&=&-\ell \alpha \omega_1+(\omega_0-\omega_1)\mathbf{1}\left\{true(k)=0\right\}\ell(1-\alpha)\nonumber\\
			&=&-\ell(\omega_0-\omega_1)(1-\alpha)\left(\frac{\alpha}{1-\alpha}\cdot\frac{\omega_1}{\omega_0-\omega_1} -\mathbf{1}\left\{true(k)=0\right\}\right)\nonumber\\
			&=&-\ell(\omega_0-\omega_1)(1-\alpha)\left(\beta-\mathbf{1}\left\{true(k)=0\right\}\right).\nonumber
		\end{eqnarray}
		In expectation, $$-\ell(\omega_0-\omega_1)(1-\alpha)\left(\beta-p_k\right)\leq 0$$ since $p_k\leq\beta$ by assumption. The resulting schedule is consistent with the $\beta$-threshold rule from time $t$ onward. This establishes the desired contradiction and concludes the proof.\qedhere
	\end{enumerate}
\end{proof}

Given our results in Theorem \ref{thm:beta-optimal}, we now provide a technical reason behind Assumption \ref{asm:wsrpt} which requires $\omega_1 < \omega_0 (1-\alpha)$. Suppose on the contrary that $\omega_1\geq \omega_0 (1-\alpha)$. Rearranging inequalities, this also implies that $\alpha\geq 1-\omega_1/\omega_0$. Then, $$\beta=\frac{\alpha}{1-\alpha}\cdot\frac{\omega_1}{\omega_0-\omega_1}\geq \frac{\alpha\omega_0}{\omega_0-\omega_1} = \frac{\alpha}{1-\omega_1/\omega_0}\geq 1$$ and so by the $\beta$-threshold rule we would complete every job nonpreemptively. Because this is not a particularly interesting case, we focus our efforts where preemption offers room for improvement. All the same, we provide a performance upper bound for this case in Corollary \ref{cor:asm-guarantee}.

The $\beta$-threshold rule may seem arbitrary at first, but there is an intuitive explanation for it that reveals a strong connection with the celebrated $c\mu$ rule. Recall that for any job $j$, $\mathbb{E}(w_j) = \omega_1+(\omega_0-\omega_1)p_j$ by Equation \eqref{eq:job-weight}.

\begin{proposition}
$$\frac{\mathbb{E}(w_j)}{1}>\frac{\omega_1}{1-\alpha}\iff p_j>\beta.$$
\end{proposition}
\begin{proof}
	Expanding the left hand side of the inequality,
	\begin{align*}
		(1-\alpha)\left(\omega_1+(\omega_0-\omega_1)p_j\right)>\omega_1 &\iff (1-\alpha)(\omega_0-\omega_1)p_j >\alpha\omega_1\\
		&\iff p_j>\frac{\alpha}{1-\alpha}\cdot\frac{\omega_1}{\omega_0-\omega_1}\\
		&\iff p_j>\beta,
	\end{align*}
	which is equivalent to the conditions given in the $\beta$-threshold rule. Inequality in the other direction holds analogously.
\end{proof}

At every decision point, applying the $\beta$-threshold rule is equivalent to comparing the $c\mu$ of an unopened job $j$ against the $c\mu$ of a known low priority job with $1-\alpha$ units of residual work, and choosing the job with the higher $c\mu$ value.


Depending on our chosen parameter values, the $\beta$-threshold rule may give rise to three modes of decision making: a \textit{nonpreemptive} policy, a \textit{preemptive} policy, and a \textit{hybrid} policy that switches from a preemptive policy to a nonpreemptive policy sometime in between.

Let us first assume that jobs are sorted in WSPT order of predicted priorities. A \textit{nonpreemptive} policy completes every job in sorted order without preemption. A schedule generated by a nonpreemptive policy is a nonpreemptive schedule. A policy is \textit{preemptive} if, opening jobs in sorted order, every type 1 job is preempted at its $\alpha$-point. These low priority jobs will only be revisited once all $n$ jobs have been opened and every high priority job has completed its processing. The resulting schedule is a preemptive schedule. Preemptive and nonpreemptive schedules are two non-adaptive special cases of a schedule generated by the $\beta$-threshold rule.

The \textit{hybrid} policy, on the other hand, is an adaptive policy that switches between the preemptive and nonpreemptive regimes based on the predictive label of the job being processed. More specifically, a preemptive strategy is used on jobs that are expected to be type 0, while a nonpreemptive strategy is used on the remaining jobs that are predicted to be non-urgent. Given our initial sort, we make this switch exactly once.

In what follows, we specify the conditions that give rise to each of our policies.

\begin{corollary}\label{cor:beta-cond-np}
	A nonpreemptive policy is optimal if $$\min\left(\rho(1-\beta),\beta(1-\rho)\right) \leq \rho(1-\beta)\varepsilon_0+\beta(1-\rho)\varepsilon_1\text{  and  }\rho\leq\beta.$$
\end{corollary}
\begin{proof}
	The statement follows directly from Theorem \ref{thm:beta-optimal}. We employ a nonpreemptive policy if $\beta \geq\mathbb{P}(true(\cdot)=0|pred(\cdot)=0)$ where the probability is as defined in \eqref{eq:pj-1}-\eqref{eq:pj-2}. Rearranging the inequality, we obtain the result.
\end{proof}

\begin{corollary}\label{cor:beta-cond-p}
	A preemptive policy is optimal if $$\min\left(\rho(1-\beta),\beta(1-\rho)\right) \leq \rho(1-\beta)\varepsilon_0+\beta(1-\rho)\varepsilon_1\text{  and  } \rho>\beta.$$
\end{corollary}

\begin{corollary}\label{cor:beta-cond-hyb}
	A hybrid policy is optimal if $$\rho(1-\beta)\varepsilon_0+\beta(1-\rho)\varepsilon_1 < \min\left(\rho(1-\beta),\beta(1-\rho)\right).$$
\end{corollary}

The $\beta$-threshold rule admits a hybrid policy if 
\begin{equation}
	\mathbb{P}(true(\cdot)=0|pred(\cdot)=1)\leq\beta<\mathbb{P}(true(\cdot)=0|pred(\cdot)=0).\label{eq:hybrid-condition}
\end{equation}
It follows naturally from Bayes' rule that the gap between the two conditional probabilities in \eqref{eq:hybrid-condition} is large when prediction error is low. When that is the case, $\beta$ is much more likely to fall in between these two probabilities for our chosen parameter values, resulting in an adaptive hybrid policy. On the other hand, when we have a predictor with high prediction error, this conditional probability gap is likely to be smaller, in which case a non-adaptive policy would be best.

\section{Analysis of the $\beta$-Threshold Rule}

\subsection{Performance Analysis}

We now quantify the performance of our policies as a function of prediction error. More specifically, we fix the number of urgent jobs among our $n$ available jobs, then obtain exact expressions for expected performance conditional on this quantity, which we denote $n_0$. Performance is measured against the offline optimum $\mathsf{OPT}$ given $n_0$. This focus on conditional expectation allows us to remove one layer of randomness from our problem and isolate the effects of misprediction. The expressions we derive in this section will also be useful for competitive analysis in our next section. Extending our results to obtain expressions for unconditional expectations of performance can be easily done by using the first and second moments of $n_0$.

\begin{proposition}
	Let $C_j^\phi$ and $C_j^\alpha$ each denote job $j$'s completion time in nonpreemptive and preemptive schedules, respectively. Given $n_0$,
	\begin{align}
		\mathsf{OPT} &= (\omega_0-\omega_1)\frac{n_0(n_0+1)}{2}+\omega_1\frac{n(n+1)}{2} \label{eq:obj-offline}\\
		\mathbb{E}\left(\sum_{j=1}^n w_jC_j^\phi\middle|n_0\right) &= \mathsf{OPT} + (\omega_0-\omega_1)\mathbb{E}(X|n_0)\label{eq:obj-nonpmtv}\\
		\mathbb{E}\left(\sum_{j=1}^n w_jC_j^{\alpha}\middle|n_0\right) &= \mathsf{OPT} +  \alpha \omega_0 \mathbb{E}(X|n_0)+\alpha\omega_1 \mathbb{E}(Y|n_0) \label{eq:obj-fullpmtv}
	\end{align}
	 where, letting $n_1=n-n_0$,
	 \begin{align}
	 	\mathbb{E}(X|n_0)=\frac{(\varepsilon_0+\varepsilon_1)n_0n_1}{2},\ \mathbb{E}(Y|n_0)=\frac{n_1(n_1-1)}{2}.\label{eq:pairwise_exp}
	 \end{align}
	 \label{prop:base-perf}
\end{proposition}

\begin{proof}
The offline optimum is easy to compute by WSPT:
\begin{align*}
	\mathsf{OPT}=\sum_{j=1}^n j w_j &= \sum_{j=1}^{n_0}\omega_0 j+\sum_{j=n_0+1}^n \omega_1 j\\
		&= (\omega_0-\omega_1)\frac{n_0(n_0+1)}{2}+\omega_1\frac{n(n+1)}{2}.
\end{align*}

For (\ref{eq:obj-nonpmtv}) and (\ref{eq:obj-fullpmtv}), recall that sorting jobs in WSPT order of predicted priorities results in a number of possible permutations of true priorities. We evaluate the objective for some fixed permutation $\pi$ of true types, then take the expectation across all possible permutations.

In a nonpreemptive schedule, each pair of jobs whose true types are out of order, i.e., a pair of (true 1, true 0), adds $\omega_0-\omega_1$ to the objective relative to the offline optimum. Letting $X$ denote the number of such inversions in $\pi$, $\sum_{j=1}^n w_jC_j^\phi = \mathsf{OPT} + (\omega_0-\omega_1)X$.

In a preemptive schedule, the cost of one inversion is $\alpha\omega_0$, since a type 1 job preempts after processing $\alpha$ units and allows a type 0 job to be processed and completed before resuming its $1-\alpha$ units of residual work. This schedule also incurs a cost of $\alpha\omega_1$ for every pair of (true 1, true 1) jobs, because the policy requires that we open and preempt both jobs before we begin processing any remaining work for completion. Thus, letting $Y$ denote the number of (true 1, true 1) pairs in $\pi$, $\sum_{j=1}^n w_jC_j^\alpha = \mathsf{OPT} + \alpha\omega_0 X+\alpha\omega_1 Y$.

We conclude the proof by computing $\mathbb{E}(X|n_0)$.
\begin{align*}
	X&= \sum_{j=1}^{n}\sum_{k>j}\mathbbm{1}\left\{\pi(j)=1\right\}\mathbbm{1}\left\{\pi(k)=0\right\}\\
	&= \sum_{j=1}^{n}\sum_{k>j}\mathbbm{1}\left\{\pi(j)=1,pred(j)=0\right\}\mathbbm{1}\left\{\pi(k)=0,pred(k)=0\right\}\\
	&\qquad+\sum_{j=1}^{n}\sum_{k>j}\mathbbm{1}\left\{\pi(j)=1,pred(j)=0\right\}\mathbbm{1}\left\{\pi(k)=0,pred(k)=1\right\}\\
	&\qquad+\sum_{j=1}^{n}\sum_{k>j}\cancel{\mathbbm{1}\left\{\pi(j)=1,pred(j)=1\right\}\mathbbm{1}\left\{\pi(k)=0,pred(k)=0\right\}}\\
	&\qquad+\sum_{j=1}^{n}\sum_{k>j}\mathbbm{1}\left\{\pi(j)=1,pred(j)=1\right\}\mathbbm{1}\left\{\pi(k)=0,pred(k)=1\right\}
\end{align*}
where the third term cancels because of our initial sort in WSPT order of predicted priorities. Accounting for the order of jobs, we can replace $\pi(\cdot)$ with $true(\cdot)$:
\begin{align*}
	X&= \frac{1}{2}\left(\sum_{j=1}^{n}\mathbbm{1}\left\{true(j)=1,pred(j)=0\right\}\right)\left(\sum_{k=1}^n\mathbbm{1}\left\{true(k)=0,pred(k)=0\right\}\right)\\
	&\qquad+\left(\sum_{j=1}^{n}\mathbbm{1}\left\{true(j)=1,pred(j)=0\right\}\right)\left(\sum_{k=1}^n\mathbbm{1}\left\{true(k)=0,pred(k)=1\right\}\right)\\
	&\qquad+\frac{1}{2}\left(\sum_{j=1}^{n}\mathbbm{1}\left\{true(j)=1,pred(j)=1\right\}\right)\left(\sum_{k=1}^n\mathbbm{1}\left\{true(k)=0,pred(k)=1\right\}\right).
\end{align*}
Order in the second term is, again, automatically satisfied by how we sort the jobs. Taking the conditional expectation given $n_0$ and letting $n_1=n-n_0$, we obtain
\begin{align}
	\mathbb{E}(X|n_0) &= \frac{\varepsilon_1(1-\varepsilon_0)n_0n_1}{2}+\varepsilon_0\varepsilon_1 n_0 n_1 + \frac{\varepsilon_0(1-\varepsilon_1)n_0n_1}{2}\label{eq:condEX}\\
	&=\frac{(\varepsilon_0+\varepsilon_1) n_0n_1}{2}.\nonumber
\end{align}
Finally, $\mathbb{E}(Y|n_0) = {n_1\choose 2}=n_1(n_1-1)/2$.
\end{proof}

Given the proposition above, we can combine (\ref{eq:obj-nonpmtv})-(\ref{eq:obj-fullpmtv}) to give expressions for the performance of the $\beta$-threshold rule. In essence, the $\beta$-threshold rule dictates when to move from a preemptive regime to a nonpreemptive regime. Based on our previous analyses, this cutoff occurs once we complete the last job that is predicted to be of high priority.

\begin{proposition}
	Let $C_j^\beta$ denote job $j$'s completion time in a schedule generated by the $\beta$-threshold rule. This schedule is nonpreemptive with performance given in (\ref{eq:obj-nonpmtv}) if $\beta\geq\max_jp_j$, and preemptive with performance given in (\ref{eq:obj-fullpmtv}) if $\beta<\min_jp_j$. Otherwise, relative to $\mathsf{OPT}$ as defined in (\ref{eq:obj-offline}), the conditional expectation given $n_0$ is 
	\begin{align}
		\mathbb{E}\left(\sum_{j=1}^n w_jC_j^\beta\middle|n_0\right)=\mathsf{OPT} + \mathbb{E}\left(\underbrace{\alpha \omega_0 X_0+\alpha\omega_1 Y_0}_\text{preemptive}+\underbrace{(\omega_0-\omega_1)(X-X_0)}_\text{nonpreemptive}\middle|n_0\right)\label{eq:obj-beta}
	\end{align}
	where, letting $n_1=n-n_0$,
	\begin{align*}
		\mathbb{E}(X_0|n_0)=\frac{\varepsilon_1(1-\varepsilon_0)n_0n_1}{2},\ \mathbb{E}(Y_0|n_0)=\frac{\varepsilon_1^2\left(n_1^2-n_1\right)}{2}
	\end{align*}
	and $\mathbb{E}(X|n_0)$ is as defined in (\ref{eq:pairwise_exp}).\label{prop:beta-perf}
\end{proposition}
\begin{proof}
	$X_0$ and $Y_0$ count the number of (true 1, true 0) and (true 1, true 1) pairs, respectively, from the set of jobs that are predicted to be of type 0. The expected value of $X_0$ given $n_0$ is given in the first term of \eqref{eq:condEX} in Proposition \ref{prop:base-perf}. The expected value of $Y_0$ given $n_0$ is the expected value of ${\mathsf{Bin}(n_1,\varepsilon_1)\choose 2}$ where $\mathsf{Bin}(n_1,\varepsilon_1)$ denotes a binomial random variable with parameters $n_1$ and $\varepsilon_1$. The remainder of the proof is identical to the one given in the proposition above.
\end{proof}

The expression in \eqref{eq:obj-beta} makes it clear that the impacts of false positive and false negative rates to performance may vary.

\begin{corollary}
	If the false positive rate $\varepsilon_1=0$, a hybrid policy gives a nonpreemptive schedule.
\end{corollary}

\begin{figure}[h]
	\centering
	\includegraphics[scale =.55]{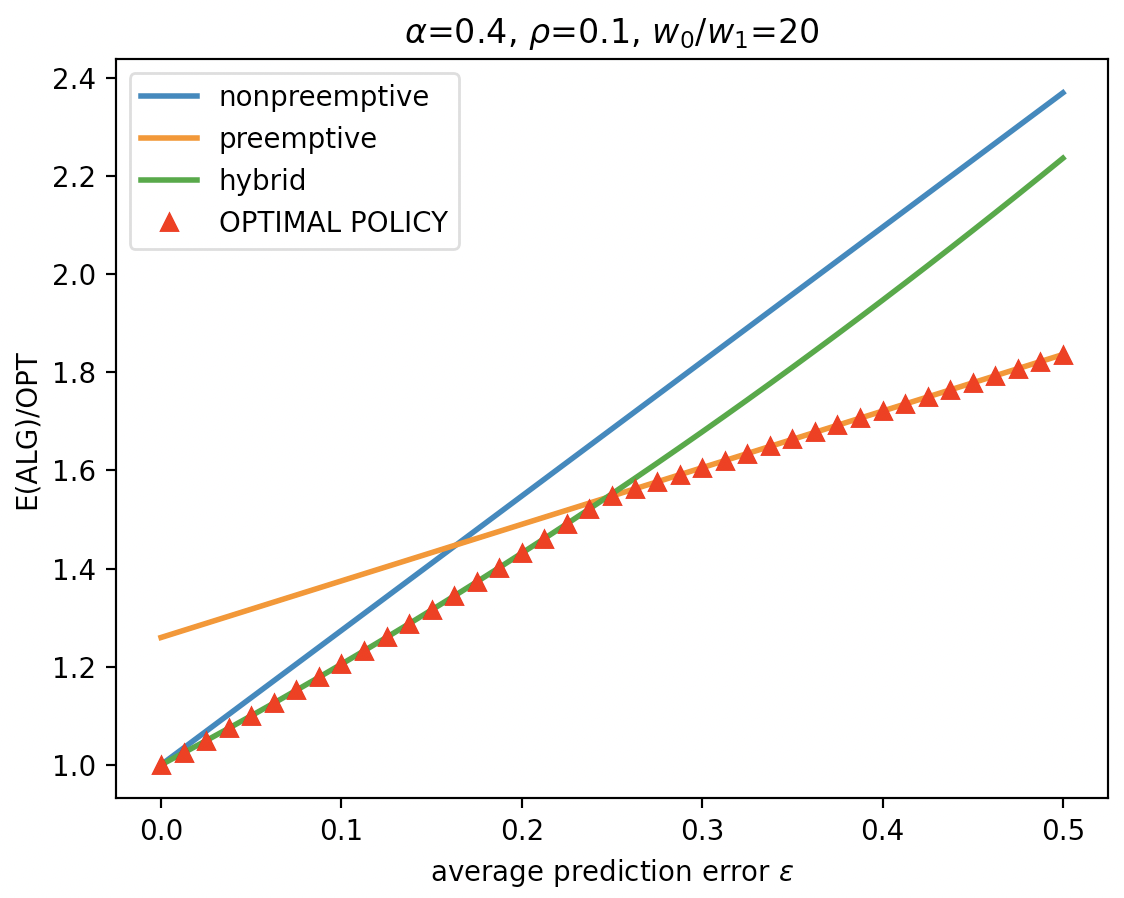}
	\caption{Expected performance}
	\label{fig:expperf_avgcase}
\end{figure}

Figure \ref{fig:expperf_avgcase} plots the unconditional expected performance of each of our policies as a function of prediction error, where performance is normalized by the offline optimum. For illustrative purposes, we assume $\varepsilon_0=\varepsilon_1$ and choose parameter values of $\alpha=0.4$, $\rho=0.1$, and $\omega_0/\omega_1=20$. Since our problem is a minimization problem, the lower the ratio of $\mathbb{E}(\mathsf{ALG})/\mathsf{OPT}$, the better.

The nonpreemptive policy performs very well when prediction error is low, in fact recovering the offline optimum when we are given perfect predictions. This policy blindly trusts the predictor, however, resulting in poor performance when prediction quality is low. On the other hand, the preemptive policy opts \textit{not} to trust the predictions and searches for high priority jobs regardless of the advice it receives. It performs well when prediction quality is low, but is overly aggressive against non-urgent jobs when predictions are accurate, penalizing them unnecessarily. The $\beta$-threshold rule takes the best of both worlds. When prediction error is low, our optimal policy strategically shifts from a preemptive to a nonpreemptive policy, outperforming each of the individual non-adaptive policies. Once prediction error reaches a certain point and predictive labels lose meaning, the $\beta$-threshold rule shifts to a preemptive policy.

\begin{figure}[h]
	\centering
	\includegraphics[scale =.4]{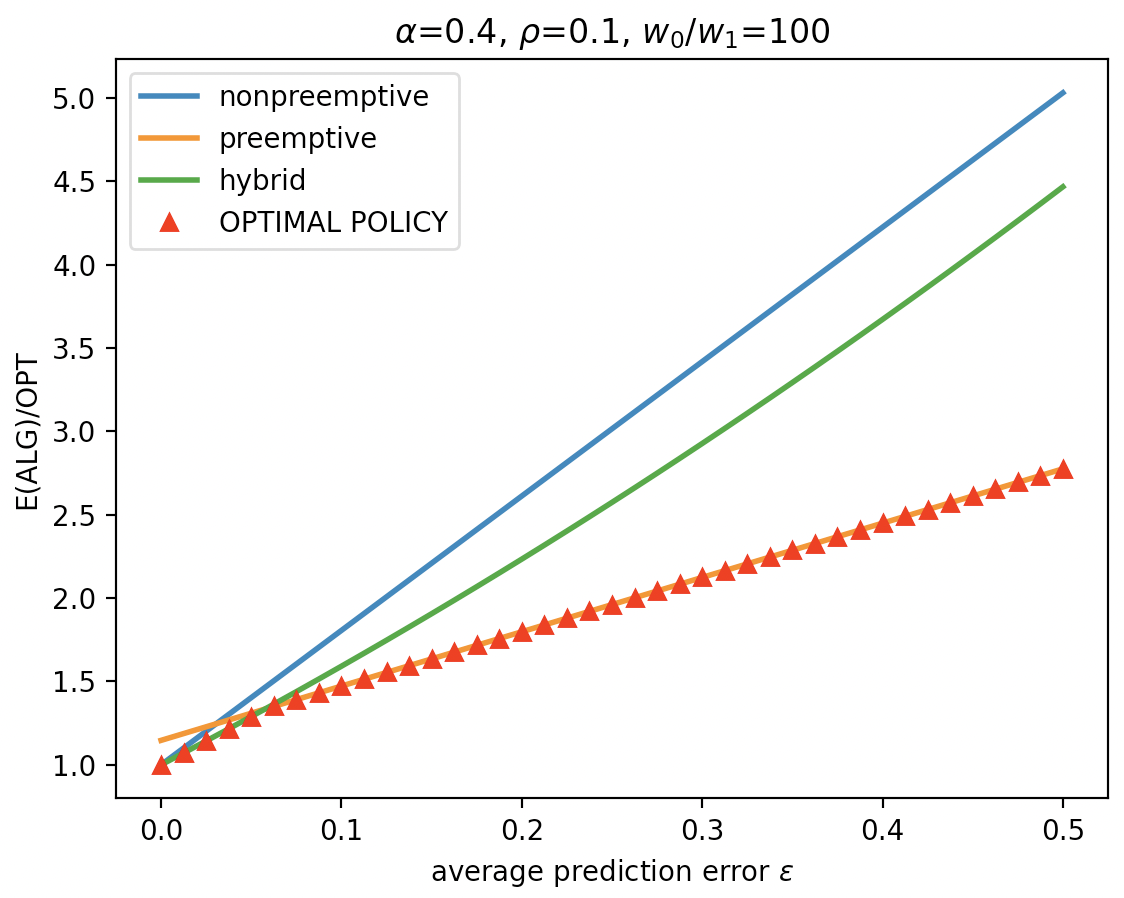}~\includegraphics[scale =.4]{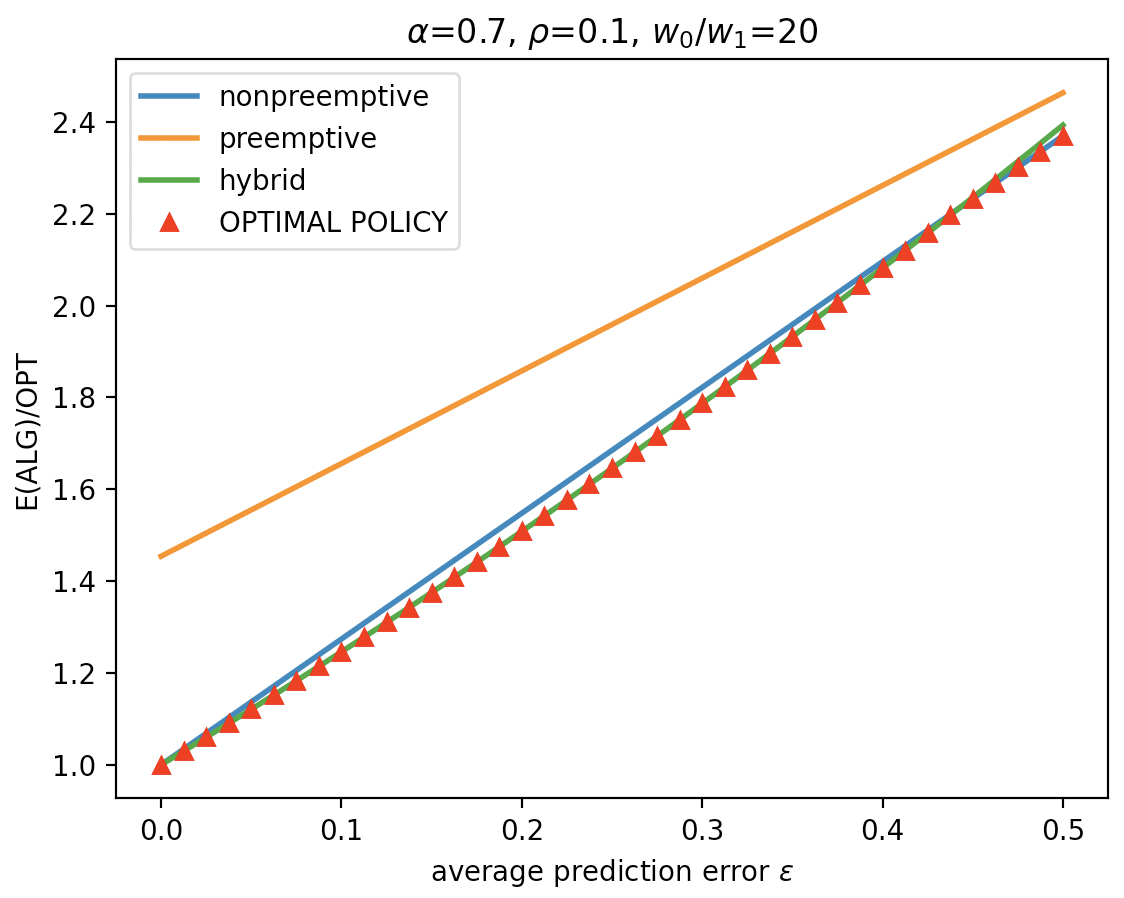}
	\caption{Expected performance}
	\label{fig:expperf-extreme}
\end{figure}

Needless to say, performance depends heavily on our chosen parameter values. Figure \ref{fig:expperf-extreme} gives two examples in which performance improvements from the $\beta$-threshold rule are modest at best. The plot on the left panel considers a case where relative priority values are set very high at $\omega_0/\omega_1=100$. Analytically, our chosen parameter values push down the $\beta$ value significantly so that it becomes unlikely that $\beta$ will fall between the two conditional probabilities given in \eqref{eq:hybrid-condition} unless the predictor is very accurate. Intuitively, the relative priority of urgent jobs is so great that there is simply no room for prediction error. This explains the low tolerance for error before our optimal policy switches from a hybrid policy to a preemptive policy. The hybrid policy still outperforms both non-adaptive policies when predictions are accurate.

The second plot in Figure \ref{fig:expperf-extreme} is a rare example in which a nonpreemptive policy outperforms a preemptive policy throughout. Here, we consider a high value of $\alpha$ where $\alpha=0.7$, which pushes up the $\beta$ value and the cost of preemption at the same time. In this case, the high cost of preemption makes it preferable to complete a low priority job than to open a new job that may or may not be high priority. We again observe that the hybrid policy outperforms both non-adaptive policies, but the improvements are small.

\subsection{Competitive Analysis}

We provide performance guarantees for our policies in this section. 

\begin{definition}
	The \emph{competitive ratio} of an online algorithm $\mathsf{ALG}$ is $\mathsf{CR}$ if the inequality $$\mathbb{E}(\mathsf{ALG})\leq\mathsf{CR}\cdot\mathsf{OPT}$$ holds for all possible inputs. Then, we can also say that $\mathsf{ALG}$ is $\mathsf{CR}$-\emph{competitive}.
\end{definition}

In our competitive analysis, an adversary deliberately choosing a difficult input has control over the mix of urgent and non-urgent jobs. Let $q$ denote the fraction of urgent jobs among all available jobs. We shall aim to find the worst case values of $q$. This analysis addresses a known weakness in our model. Our model assumes that each arriving job is independently a high priority job with probability $\rho\in (0,1)$, and our proposed $\beta$-threshold rule is optimal with respect to this parameter. While $\rho$ could be inferred from historical data, it also tends to be highly volatile and sensitive to environmental changes. Mass casualty events or insurance policy changes are some examples that could drive the value of $\rho$ up or down. With competitive analysis, we are able to guarantee performance for all possible values of $\rho$. Furthermore, as a byproduct of our analyses, we can observe which values of $\rho$ result in the worst case.

\begin{lemma}
	The performance of a nonpreemptive policy is bounded above by $\mathsf{CR}^\phi\cdot\mathsf{OPT}$ where $$\mathsf{CR}^\phi=1+\varepsilon\left(\sqrt{\frac{\omega_0}{\omega_1}}-1\right)$$ and $\varepsilon=(\varepsilon_0+\varepsilon_1)/2$ is the average of the false negative and false positive rates $\varepsilon_0$ and $\varepsilon_1$. \label{lemma:cr-np}
\end{lemma}
\begin{proof}
	Let $C_j^\phi$ denote job $j$'s completion time in a nonpreemptive schedule. Based on (\ref{eq:obj-nonpmtv}), $$\frac{\mathbb{E}\left(\sum_{j=1}^n w_jC_j^\phi\middle|n_0\right)}{\mathsf{OPT}}-1=\frac{(\omega_0-\omega_1)\mathbb{E}(X|n_0)}{\mathsf{OPT}}.$$ Expanding the terms as given in Proposition \ref{prop:base-perf} and letting $q=n_0/n$,
	\begin{align*}
		\frac{(\omega_0-\omega_1)\mathbb{E}(X|n_0)}{\mathsf{OPT}} &= \frac{2\varepsilon(\omega_0-\omega_1) n_0 (n-n_0)}{(\omega_0-\omega_1) n_0(n_0+1)+\omega_1 n(n+1)}\\
		&= \frac{2\varepsilon(\omega_0-\omega_1) q(1-q)n^2}{(\omega_0-\omega_1) (q^2n^2+qn)+\omega_1 (n^2+n)}
	\end{align*} then, we approach the limit from below as $n\rightarrow\infty$ so the upper bound is
	\begin{align*}
		\frac{2\varepsilon(\omega_0-\omega_1)q(1-q)}{(\omega_0-\omega_1) q^2+\omega_1}.
	\end{align*}
	
	It is straightforward calculus to show that $$\max_{0\leq q\leq 1}\ \frac{2\varepsilon(\omega_0-\omega_1)q(1-q)}{(\omega_0-\omega_1) q^2+\omega_1}= \varepsilon\left(\sqrt{\frac{\omega_0}{\omega_1}}-1\right),$$ which implies the result. The maximum is attained by choosing $$q=\sqrt{\frac{\omega_1}{\omega_0-\omega_1}+\left(\frac{\omega_1}{\omega_0-\omega_1}\right)^2}-\frac{\omega_1}{\omega_0-\omega_1}.$$
\end{proof}

An immediate consequence of the above lemma is a performance bound for the schedule when Assumption \ref{asm:wsrpt} does not hold. Recall that without Assumption \ref{asm:wsrpt}, a nonpreemptive schedule is an optimal schedule.

\begin{corollary}
	If $\omega_1\geq \omega_0 (1-\alpha)$, the competitive ratio is $1+\varepsilon\left(\sqrt{\frac{1}{1-\alpha}}-1\right)$. \label{cor:asm-guarantee}
\end{corollary}

\begin{lemma}
	The competitive ratio $\mathsf{CR}^\alpha$ of a preemptive policy is
	$$\mathsf{CR}^\alpha=\begin{dcases*}
		1+\alpha & if $\varepsilon\leq\omega_1/\omega_0$, and\\
		1+\frac{\alpha}{2}\frac{\omega_0}{\omega_0-\omega_1}\left(1-2\varepsilon+\sqrt{1-4\varepsilon+4\varepsilon^2\left(\frac{\omega_0}{\omega_1}\right)}\right) & otherwise,	
	\end{dcases*}$$ where $\varepsilon=(\varepsilon_0+\varepsilon_1)/2$ is the average of the false negative and false positive rates $\varepsilon_0$ and $\varepsilon_1$. \label{lemma:cr-p}
\end{lemma}
\begin{proof}
	The first part of the proof proceeds similarly. Let $C_j^\alpha$ denote job $j$'s completion time in a preemptive schedule. Letting $n_1=n-n_0$ and $q=n_0/n$,
	\begin{align*}
		\frac{\mathbb{E}\left(\sum_{j=1}^n w_jC_j^\alpha\middle|n_0\right)}{\mathsf{OPT}}-1&=\frac{\alpha \omega_0 \mathbb{E}(X|n_0)+\alpha\omega_1 \mathbb{E}(Y|n_0)}{\mathsf{OPT}}\\
		&= \alpha\cdot\frac{2\varepsilon\omega_0 n_0n_1+\omega_1n_1(n_1-1)}{(\omega_0-\omega_1) n_0(n_0+1)+\omega_1 n(n+1)}\\
		&= \alpha\cdot \frac{2\varepsilon\omega_0 q(1-q)n^2+\omega_1 \left((1-q)^2n^2-(1-q)n\right)}{(\omega_0-\omega_1) \left(q^2n^2+qn\right)+\omega_1 (n^2+n)}
	\end{align*}
	then, we approach the limit from below as $n\rightarrow\infty$ so the upper bound is
	\begin{align}
		\alpha\cdot \frac{2\varepsilon\omega_0 q(1-q)+\omega_1 (1-q)^2}{(\omega_0-\omega_1) q^2+\omega_1}.\label{eq:deriv}
	\end{align}
	
	We want to maximize \eqref{eq:deriv} with respect to $q$ where $0\leq q\leq 1$. Taking the derivative,
	\begin{eqnarray*}
		&&\alpha\cdot\frac{\left(\splitfrac{\left(2\varepsilon\omega_0(1-2q)-2\omega_1(1-q)\right)\left((\omega_0-\omega_1) q^2+\omega_1\right)}{-2(\omega_0-\omega_1)q\left(2\varepsilon\omega_0 q(1-q)+\omega_1 (1-q)^2\right)}\right)}{\left((\omega_0-\omega_1) q^2+\omega_1\right)^2}\\
		&=&\alpha\cdot\frac{2(\omega_0-\omega_1)(\omega_1-\varepsilon\omega_0)q^2+2\omega_1(2(\omega_1-\varepsilon\omega_0)-\omega_0)q-2\omega_1(\omega_1-\varepsilon\omega_0)}{\left((\omega_0-\omega_1) q^2+\omega_1\right)^2}.
	\end{eqnarray*}
	We first evaluate this derivative at the boundaries. At $q=1$, the numerator is always non-positive with $-2\varepsilon\omega_0^2\leq 0$. When $q=0$, the numerator is $-2\omega_1(\omega_1-\varepsilon\omega_0)$. If this quantity is non-positive, i.e., $\varepsilon\leq \omega_1/\omega_0$, then the coefficient for $q^2$ given by $2(\omega_0-\omega_1)(\omega_1-\varepsilon\omega_0)$ is also non-negative, which implies that \eqref{eq:deriv} decreases in $q$ everywhere in the domain $0\leq q\leq 1$. Thus, if $\varepsilon\leq \omega_1/\omega_0$, we obtain the competitive ratio $1+\alpha$ by setting $q=0$ in \eqref{eq:deriv}.
	
	If $\varepsilon > \omega_1/\omega_0$, the expression in \eqref{eq:deriv} attains a maximum in the interior of the domain. The rest of the proof is straightforward calculus. We achieve the maximum given in the statement of the lemma by setting $$q=\sqrt{\frac{\omega_1}{\omega_0-\omega_1}+\left(\frac{\omega_1}{\omega_0-\omega_1}\cdot\frac{2\varepsilon\omega_0-2\omega_1+\omega_0}{2\varepsilon\omega_0-2\omega_1}\right)^2}-\frac{\omega_1}{\omega_0-\omega_1}\cdot\frac{2\varepsilon\omega_0-2\omega_1+\omega_0}{2\varepsilon\omega_0-2\omega_1}.$$
\end{proof}

A preemptive policy aggressively searches for high priority jobs by preempting every type 1 job it encounters, completing any low priority residual work only after all jobs are open and all type 0 jobs have completed their processing. Lemma \ref{lemma:cr-p} confirms our intuition that this policy performs poorly when prediction error is low. Consider for example an instance that consists exclusively of type 1 jobs. Indiscriminate preemption offers no advantage, as there are no urgent jobs to search for. In this case, a preemptive policy causes on average an $\alpha$ unit of delay in the completion of every job, resulting in a constant $1+\alpha$ competitive ratio when error rates are small ($\varepsilon\leq\omega_1/\omega_0$). Given our assumption that prediction errors are at most one half, we are able to deduce the following corollary.

\begin{corollary}
	If $\omega_0 < 2\omega_1$, a fully preemptive policy is $(1+\alpha)$-competitive.
\end{corollary}

Thus, preemption offers little advantage when the relative weight differential is small.

\begin{lemma}
	A hybrid policy achieves a competitive ratio of
	\begin{equation}
		\mathsf{CR}^\beta := 1+\frac{1}{2}\left( \alpha\varepsilon_1^2-\lambda+\sqrt{\frac{\omega_0}{\omega_1}\lambda^2+\frac{\omega_0}{\omega_0-\omega_1}\left(\alpha\varepsilon_1^2\right)^2} \right)\label{eq:cr-hyb-exact}
	\end{equation}
	where
	\begin{equation}
		\lambda = \varepsilon_0(1+\varepsilon_1)+\frac{\alpha\omega_0}{\omega_0-\omega_1}\varepsilon_1(1-\varepsilon_0)-\frac{\alpha\omega_1}{\omega_0-\omega_1}\varepsilon_1^2.\label{eq:lambda-crbeta}
	\end{equation}\label{lemma:cr-hyb}
\end{lemma}
\begin{proof}
	The proof proceeds similarly where $C_j^\beta$ denotes job $j$'s completion time in a hybrid policy. Letting $n_1=n-n_0$ and $q=n_0/n$,
	\begin{eqnarray*}
		&&\frac{\mathbb{E}\left(\sum_{j=1}^n w_jC_j^\beta\middle|n_0\right)}{\mathsf{OPT}}-1\\
		&=&\frac{\mathbb{E}\left(\alpha \omega_0 X_0+\alpha\omega_1 Y_0+(\omega_0-\omega_1)(X-X_0)\middle|n_0\right)}{\mathsf{OPT}}\\
		&=&\frac{\left(\splitfrac{\alpha\omega_0\varepsilon_1(1-\varepsilon_0)q(1-q)n^2+\alpha\omega_1 \varepsilon_1^2\left((1-q)^2n^2-(1-q)n\right)}{+(\omega_0-\omega_1)\varepsilon_0(1+\varepsilon_1)q(1-q)n^2}\right)}{(\omega_0-\omega_1) \left(q^2n^2+qn\right)+\omega_1 (n^2+n)}
	\end{eqnarray*}
	then, we approach the limit from below as $n\rightarrow\infty$ so the upper bound is
	\begin{align}
		\frac{\left(\alpha\omega_0\varepsilon_1(1-\varepsilon_0)+(\omega_0-\omega_1)\varepsilon_0(1+\varepsilon_1)\right)q(1-q)+\alpha\omega_1 \varepsilon_1^2(1-q)^2}{(\omega_0-\omega_1) q^2+\omega_1}.\label{eq:deriv-beta}
	\end{align}
	Using arguments similar to those given in the previous lemma, (\ref{eq:deriv-beta}) always attains a maximum in the domain $0\leq q\leq 1$ when $$q=\sqrt{\frac{\omega_1}{\omega_0-\omega_1}+\left(\frac{\omega_1}{\omega_0-\omega_1}\cdot\frac{\lambda+\alpha\varepsilon_1^2}{\lambda-\alpha\left(\frac{\omega_1}{\omega_0-\omega_1}\right) \varepsilon_1^2}\right)^2}-\frac{\omega_1}{\omega_0-\omega_1}\cdot\frac{\lambda+\alpha\varepsilon_1^2}{\lambda-\alpha\left(\frac{\omega_1}{\omega_0-\omega_1}\right) \varepsilon_1^2}$$ where $\lambda$ is as defined above.
\end{proof}

An interpretable upper bound for \eqref{eq:cr-hyb-exact} can be derived using the inequality $\sqrt{a+b}\leq\sqrt{a}+\sqrt{b}$, which yields
	\begin{equation*}
		\mathsf{CR}^\beta \leq 1+ \frac{\lambda}{2}\left(\sqrt{\frac{\omega_0}{\omega_1}}-1\right)+\frac{\alpha\varepsilon_1^2}{2} \left(1+\sqrt{\frac{\omega_0}{\omega_0-\omega_1}}\right).
	\end{equation*}
	Rearranging \eqref{eq:lambda-crbeta}, we have
	\begin{equation}
		\lambda = \varepsilon_0 + \varepsilon_1\left(\alpha+(1-\alpha)\varepsilon_0+\frac{\alpha\omega_1}{\omega_0-\omega_1}(1-\varepsilon_0-\varepsilon_1)\right).\label{eq:rearranged-ld}
	\end{equation}
	We first show that $\lambda/2\leq\varepsilon$, where $\varepsilon$ is the average of $\varepsilon_0$ and $\varepsilon_1$. To do so, it suffices to show that the coefficient to $\varepsilon_1$ in \eqref{eq:rearranged-ld} is no greater than 1.
	\begin{eqnarray*}
		&&1-\left(\alpha+(1-\alpha)\varepsilon_0+\frac{\alpha\omega_1}{\omega_0-\omega_1}(1-\varepsilon_0-\varepsilon_1)\right)\\
		&=& (1-\alpha)(1-\varepsilon_0)-\frac{\alpha\omega_1}{\omega_0-\omega_1}(1-\varepsilon_0-\varepsilon_1)\\
		&=& (1-\alpha)\left(1-\varepsilon_0-\frac{\alpha}{1-\alpha}\cdot\frac{\omega_1}{\omega_0-\omega_1}(1-\varepsilon_0-\varepsilon_1)\right)\\
		&=& (1-\alpha)\left(1-\varepsilon_0-\beta(1-\varepsilon_0-\varepsilon_1)\right)\\
		&=& (1-\alpha)\left(\left(1-\beta\right)(1-\varepsilon_0)+\beta\varepsilon_1\right)\geq 0.
	\end{eqnarray*}
	The last inequality follows since every term in the expression is nonnegative, so we have the desired inequality. Recalling that the nonpreemptive competitive ratio is $\mathsf{CR}^\phi=1+\varepsilon\left(\sqrt{\omega_0/\omega_1}-1\right)$, we are able to decompose the competitive ratio as follows:
	\begin{equation}
		\mathsf{CR}^\beta \leq \underbrace{ 1+ \frac{\lambda}{2}\left(\sqrt{\frac{\omega_0}{\omega_1}}-1\right)}_{\substack{\text{$\leq\mathsf{CR}^\phi$}\\ \text{gains relative to $\mathsf{CR}^\phi$}}}+ \underbrace{\frac{\alpha\varepsilon_1^2}{2} \left(1+\sqrt{\frac{\omega_0}{\omega_0-\omega_1}}\right)}_{\text{losses relative to $\mathsf{CR}^\phi$}}.\label{eq:cr-decomposed}
	\end{equation}
	When the relative urgency $\omega_0/\omega_1 \gg 1$, gains in the $\beta$-threshold policy relative to a nonpreemptive policy are large since the multiplier $\sqrt{\omega_0/\omega_1}-1$ is large. In comparison, the losses are approximately equal to $\alpha\varepsilon_1^2$ and small, which implies a guaranteed performance improvement.

\begin{theorem}
	The $\beta$-threshold policy achieves a competitive ratio
	$$\begin{cases}
		\mathsf{CR}^\phi &  \text{if }\rho(1-\beta)\varepsilon_0+\beta(1-\rho)\varepsilon_1 \geq \min\left(\rho(1-\beta),\beta(1-\rho)\right)\text{ and }\rho\leq\beta,\\
		\mathsf{CR}^\alpha & \text{if }\rho(1-\beta)\varepsilon_0+\beta(1-\rho)\varepsilon_1 \geq \min\left(\rho(1-\beta),\beta(1-\rho)\right)\text{ and }\rho>\beta,\text{ and}\\
		\mathsf{CR}^\beta & \text{if }\rho(1-\beta)\varepsilon_0+\beta(1-\rho)\varepsilon_1 < \min\left(\rho(1-\beta),\beta(1-\rho)\right).
	\end{cases}$$\label{thm:betaCR}
\end{theorem}
\begin{proof}
	The theorem combines Lemmas \ref{lemma:cr-np}, \ref{lemma:cr-p} and \ref{lemma:cr-hyb} and the conditions in Corollaries \ref{cor:beta-cond-np}, \ref{cor:beta-cond-p} and \ref{cor:beta-cond-hyb}.
\end{proof}

An immediate observation from our competitive analyses is that the worst-case fraction of urgent jobs is inversely proportional to relative priority levels $\omega_0/\omega_1$. But more importantly, we are able to characterize how the competitive ratio evolves as a function of prediction error.

\begin{figure}[h]
	\centering
	\includegraphics[scale =.55]{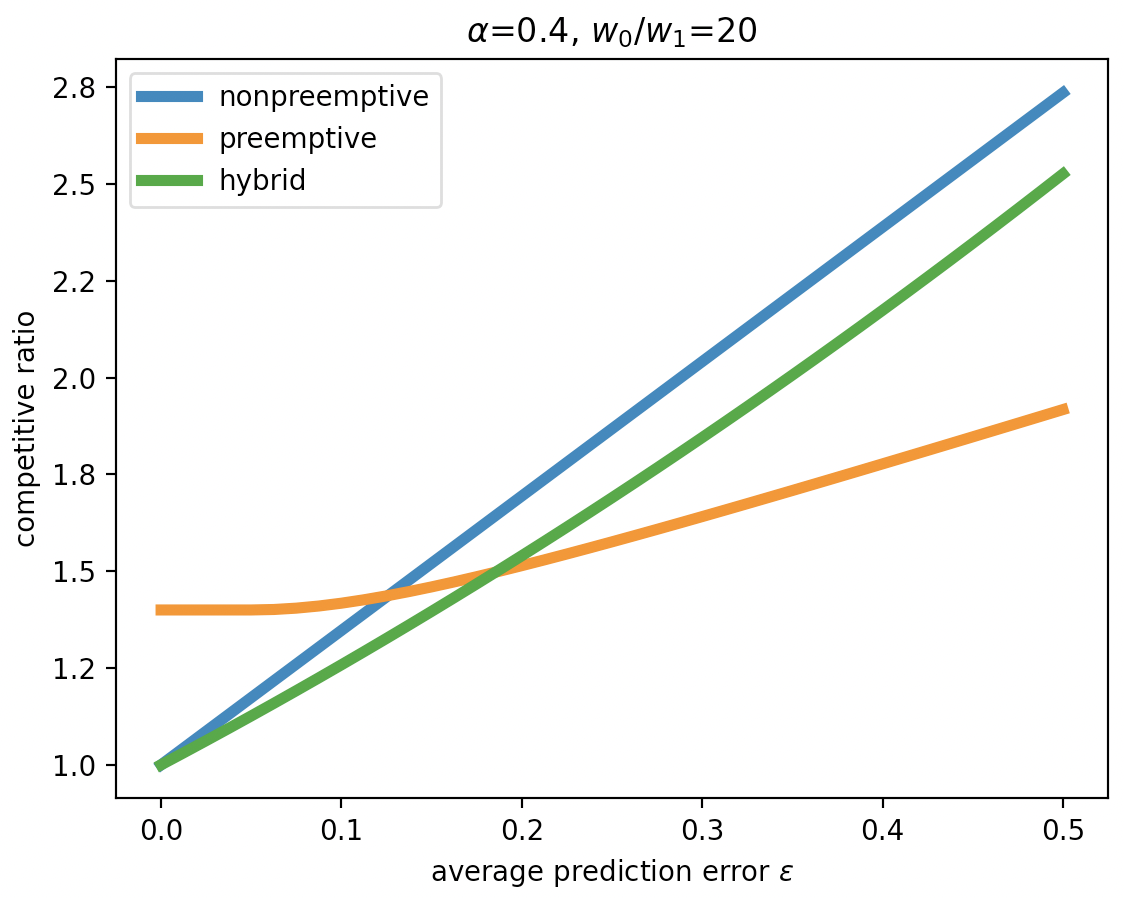}
	\caption{Competitive ratios}
	\label{fig:competitive_ratio}
\end{figure}

For analytical purposes, let us assume $\varepsilon=\varepsilon_0=\varepsilon_1$ and $\omega_0/\omega_1\gg 1$. We can easily see from our expression of $\mathsf{CR}^\phi$ in Lemma \ref{lemma:cr-np} that the competitive ratio of a nonpreemptive policy grows linearly in $O(\varepsilon)$, i.e., performance improves with prediction accuracy. In the case of a preemptive policy, the competitive ratio stays constant at $1+\alpha$ before it starts growing linearly in $O(\alpha\varepsilon)$. Compared with that of a nonpreemptive policy, its growth rate is scaled down by a factor of $\alpha$ where $0<\alpha<1$. Finally, Lemma \ref{lemma:cr-hyb} shows that the competitive ratio of a hybrid policy grows \textit{quadratically} in $O\left(\varepsilon^2\right)$. Given $\varepsilon\leq 1/2$, this rate of growth is slower than that of a nonpreemptive policy, offering yet another interpretation of the decomposition of $\mathsf{CR}^\beta$ given in \eqref{eq:cr-decomposed}. These findings are illustrated in Figure \ref{fig:competitive_ratio}.

Our analyses indicate that, for all three policies that the $\beta$-threshold rule admits, performance degrades gracefully as a function of prediction error. As such, we achieve the two qualities that an online algorithm with advice should exhibit: \textit{consistency} and \textit{robustness} \cite{lykouris21}. Consistency requires performance improvement when the predictor has low error. The idea is that performance with good advice should be better than performance with poor advice. At the same time, an algorithm should be robust to all inputs, with or without predictions. All three of our policies show improved performance with prediction accuracy. The nonpreemptive and hybrid policies even recover the offline optimum when offered perfect predictions. Our three policies are also robust in that the competitive ratios are bounded above when error rates are equal to one half. When $\varepsilon=1/2$, predictions are truly random, i.e., there are no predictions at play.

\section{Extensions}\label{sectionch2:extensions}

In this section, we consider a number of extensions to our model that more accurately reflect real-world settings.

\subsection{Probabilistic Classifiers}

We first consider a probabilistic classifier that is able to predict with what probability a job is of high priority. Rather than providing a binary predictive label of urgent vs. non-urgent, this probabilistic classifier directly offers an estimate of $p_j$. Let us denote these estimated probabilities as $\hat{p}_j$.

We first sort jobs in nonincreasing order of $\hat{p}_j$, breaking ties arbitrarily. As we proceed with our policy, jobs are opened in sorted order. Then, we have the following corollary to Theorem \ref{thm:beta-optimal}.

\begin{corollary}
	The $\beta$-threshold rule is optimal given a probabilistic classifier.
\end{corollary}

The $\beta$-threshold rule proceeds similarly even with a probabilistic classifier. A preemptive policy is applied to those jobs whose estimated probabilities lie above $\beta$, and a nonpreemptive policy is applied to those jobs whose $\hat{p}_j$ values fall below that threshold. The $\beta$-threshold rule remains optimal, minimizing the objective across all non-anticipating policies.

What differs from our original model is the measure of error. Beyond binary classification, the false negative and false positive rates $\varepsilon_0$ and $\varepsilon_1$ no longer apply. A more appropriate measure of error in this case would be the logarithmic loss function (also called the cross-entropy loss function) given as follows:
$$\eta=-\frac{1}{n}\sum_{j=1}^n (1-true(j))\cdot\log\left(\hat{p}_j\right)+true(j)\cdot\log\left(1-\hat{p}_j\right).$$
The $\beta$-threshold rule is the best possible policy for decision-making based on available information that is both imperfect and incomplete. However, its performance depends heavily on the accuracy of the classifier on hand. An exact characterization of performance as a function of the log-loss $\eta$ remains an open problem.

\subsection{Probabilistic Learning Outcome}

Our model assumes that a radiologist is always able to determine a job's true type at its $\alpha$-point with probability 1. Perhaps a more realistic model would be to leave some room for doubt. Suppose that at job $j$'s $\alpha$-point, we learn that job $j$ is an urgent job with some probability $\theta_j\in [0,1]$. Probability $\theta_j$ is a \textit{posterior} probability that offers a better likelihood of job $j$'s urgency based on $\alpha$ units of observed data.

In our notation, $\mathcal{S}$ denotes the set of unopened jobs. Every job in set $\mathcal{S}$ has one unit in remaining work, with an associated a \textit{prior} probability of being an urgent job. These prior probabilities are as defined in \eqref{eq:pj-1}-\eqref{eq:pj-2}. Let $\mathcal{I}$ denote the set of interrupted, previously preempted jobs that each have $1-\alpha$ units in residual work. Every job in $\mathcal{I}$ has an associated posterior probability. At every decision point, we make the decision of whether to open a job $k$ with the largest prior probability $p_k=\max_{j\in\mathcal{S}}p_j$, or to complete the remaining $1-\alpha$ units of work of job $i$, where job $i$ has the largest posterior probability among all jobs in $\mathcal{I}$ such that $\theta_i=\max_{j\in\mathcal{I}}\theta_j$. For this decision problem, the following modified version of the $\beta$-threshold rule is the best possible across all non-anticipating policies.

\begin{theorem}
	The $\left(\beta+\frac{\alpha}{1-\alpha}\frac{\omega_0}{\omega_0-\omega_1} \frac{\theta_i}{1-\theta_i}\right)$-threshold rule is optimal where $\theta_i=\max_{j\in\mathcal{I}}\theta_j$.
\end{theorem}
\begin{proof}
	The proof is nearly identical to the interchange argument given in Theorem \ref{thm:beta-optimal}, with small modifications. The main difference is in recognizing that the low priority job (job $i$) competing against the next unopened job (job $k$) in the proof of Theorem \ref{thm:beta-optimal} is now a low priority job with probability $1-\theta_i$, and a high priority job with probability $\theta_i$.
	
	As before, we consider the last decision point that deviates from this modified $\beta$-threshold rule. If $$p_k>\beta+\frac{\alpha}{1-\alpha}\frac{\omega_0}{\omega_0-\omega_1} \frac{\theta_i}{1-\theta_i}$$ and job $i$ is being processed at this decision point, the net change to the objective upon interchange is $$\underbrace{\left(-z_0\omega_0(1-\alpha)+\omega_1\left(z\alpha+z_0(1-\alpha)\right)\right)}_{\text{from }\eqref{eq:thm3-change1}}(1-\theta_i)+z(\omega_0\alpha) \theta_i$$ where $z$ is the number of jobs opened whose prior probabilities lie above the modified $\beta$ threshold, and $z_0$ is the number of type 0 jobs among them. The first half of this expression comes directly from our earlier proof, weighted by the probability that job $i$ is a low priority job. If job $i$ is an urgent job with weight $\omega_0$, it incurs an $\alpha$ unit of delay in completion for each of the $z$ jobs opened during interchange. Then,
	\begin{eqnarray*}
		&&\left(-z_0\omega_0(1-\alpha)+\omega_1\left(z\alpha+z_0(1-\alpha)\right)\right)(1-\theta_i)+z(\omega_0\alpha) \theta_i\\
		&=& -(\omega_0-\omega_1)(1-\alpha)(1-\theta_i)\left(z_0-\underbrace{\left(\beta+ \frac{\alpha}{1-\alpha}\frac{\omega_0}{\omega_0-\omega_1} \frac{\theta_i}{1-\theta_i}\right)}_{\text{modified $\beta$}}z\right).
	\end{eqnarray*} In expectation, the overall change to the objective is negative since each of the $z$ jobs have prior probabilities that lie above the modified $\beta$ threshold.
	
	The second case uses an identical argument. We modify \eqref{eq:thm3-change2} to account for the possibility that job $i$ is a type 0 job. Then,
	$$\underbrace{\left(-\ell \omega_1  + w_k\ell(1-\alpha)\right)}_{\text{from }\eqref{eq:thm3-change2}}(1-\theta_i)-\ell(\omega_0\alpha) \theta_i$$ and the rest of the proof proceeds similarly.
\end{proof}

We recover the original $\beta$-threshold rule when $\theta_i=0$, which is equivalent to learning that job $i$ is a non-urgent job with probability 1. This added uncertainty raises the $\beta$ threshold bar for opening new jobs to account for the possibility that job $i$ is an urgent job.


\subsection{Job Arrivals Over Time}

We had previously assumed that all jobs are available for processing at time 0. In this section, we consider the case where jobs arrive over time and are released for processing at various points in time. Each job $j$ has an associated release date $r_j\geq0$ and cannot be processed before then. At any given time, we assume no knowledge of jobs arriving in the future.

The offline version of this problem in which jobs' true types are known a priori can be written as $1|r_j,p_j=1,pmtn|\sum \{\omega_0,\omega_1\}C_j$ in the scheduling notation of \citet{grahamllr79}. The following theorem identifies an optimal policy for this offline problem.

\begin{theorem}
	The weighted shortest remaining processing time (WSRPT) rule is an optimal policy for $1|r_j,p_j=1,pmtn|\sum \{\omega_0,\omega_1\}C_j$. \label{thm:wsrpt}
\end{theorem}
\begin{proof}
	Consider an optimal schedule where job $k$ is being processed at time $t$. Suppose there exists another available job $j$ at $t$ such that
	\begin{equation}
		\frac{w_k}{x_k(t)}<\frac{w_j}{x_j(t)}\label{eq:cont-wsrpt}
	\end{equation}
	where $x_j(t)$ denotes the amount of work remaining in job $j$ at time $t$. If $w_j=w_k$, the optimal policy is in violation of the SRPT rule so we immediately have a contradiction \cite{schrage68}. We therefore assume that $w_j\neq w_k$. In addition, we also assume without loss of generality that job $j$ is the job with the largest weight-to-remaining-work ratio among all available jobs at $t$. We establish a contradiction by interchange.
	
	Despite our assumption that the optimal schedule prioritizes job $k$ over job $j$ at time $t$, we do not know whether job $k$ was actually completed prior to job $j$. We let $C_j$ and $C_k$ denote the completion times of jobs $j$ and $k$ in the optimal schedule, respectively, and consider both cases.
	
	\begin{enumerate}[i]
		\item $C_j<C_k$ : we use a pairwise interchange argument similar to that used in the proof of optimality of SRPT. Starting from $t$, we take the first $x_j(t)$ units devoted to processing jobs $j$ or $k$ in the optimal schedule, and use that time to process job $j$ to completion at $\hat{C}_j$. The remaining $x_k(t)$ units of time are then used to process job $k$ with completion time $\hat{C}_k=C_k$. This interchange only affects the completion time of job $j$, and $\hat{C}_j<C_j$ by construction, so we have our desired contradiction.
		
		\item $C_k<C_j$ : in this case, we require some additional pieces that are unique to our problem with two distinct weights. We first claim that job $j$ is the only job of weight $w_j$ that is processed in the interval $[t,C_j)$. At time $t$, job $j$ has the largest weight-to-remaining-work ratio, so it would be against the SRPT rule to process any other available job of weight $w_j$ until job $j$ is complete. The same is true of any job of weight $w_j$ released in the interval $[t,C_j)$ since $x_j(t)\leq 1$ and every newly arriving job has 1 unit of remaining work. Using a similar argument for job $k$ in the interval $[t,C_k)$, we can conclude that only jobs $j$ and $k$ are processed in $[t,\min(C_k,C_j))=[t,C_k)$.
		
		It is possible, however, that other jobs are processed in $[C_k,C_j)$. By our earlier claim, only jobs of weight $w_k$ can be processed in this interval. Let $\mathcal{A}$ denote the set of jobs processed in $[C_k,C_j)$ where, for every job $\ell\in\mathcal{A}$, $w_\ell=w_k$ holds. We claim that every job $\ell\in\mathcal{A}$ satisfies
		\begin{equation}
			\frac{w_\ell}{x_\ell(t)}<\frac{w_j}{x_j(t)}.\label{eq:cont-wsrpt-add}
		\end{equation}
		For notational convenience, we shall continue to use $x_\ell(t)$ on jobs that are released after $t$, as we can simply set $x_\ell(t)=1$ without affecting the analysis. If job $\ell$ has release date $r_\ell\leq t$, then $x_\ell(t)\geq x_k(t)$ since the optimal schedule would otherwise be in violation of the SRPT rule by processing job $k$ instead of $\ell$ at time $t$. Combined with \eqref{eq:cont-wsrpt}, we obtain the inequality. The same is true if job $\ell$ has release date $r_\ell\in[t,C_j)$ since $x_k(t)\leq x_\ell(t)=1$. 
		
		Lastly, we argue that every job $\ell\in\mathcal{A}$  has completion time $C_\ell\in[C_k,C_j)$. Suppose on the contrary that there exists a job that is partially processed in $[C_k,C_j)$ that completes sometime after $C_j$. Then, shifting the time units devoted to processing this job to the end of the $[C_k,C_j)$ interval allows job $j$ to be completed earlier without affecting the completion time of any other job. Doing so strictly improves the objective and contradicts the fact that we have an optimal schedule.
		
		We finally have all the ingredients we need to proceed with the interchange. We first process job $j$ in the first $x_j(t)$ units of $[t,C_j)$, followed by jobs in $\mathcal{A}\cup \{k\}$ in the remainder of the interval $[t+x_j(t),C_j)$. Then, for each job in $\mathcal{A}\cup \{k\}$, there is a delay in completion of at most $x_j(t)$ units. Job $j$, on the other hand, completes $x_k(t)+\sum_{\ell\in\mathcal{A}}x_\ell(t)$ units earlier in the schedule. The net effect to the objective is thus bounded above by
		\begin{align*}
			&-w_j\left(x_k(t)+\sum_{\ell\in\mathcal{A}}x_\ell(t)\right)+|\mathcal{A}\cup \{k\}|w_kx_j(t)\\
			=&-w_jx_k(t)+w_kx_j(t)-w_j\left(\sum_{\ell\in\mathcal{A}}x_\ell(t)\right)+|\mathcal{A}|w_kx_j(t)\\
			=& \underbrace{-w_jx_k(t)+w_kx_j(t)}_\text{$<0$ by \eqref{eq:cont-wsrpt}}+\sum_{\ell\in\mathcal{A}}\underbrace{\left(-w_jx_\ell(t)+w_\ell x_j(t)\right)}_\text{$<0$ by \eqref{eq:cont-wsrpt-add}}<0
		\end{align*}
		which contradicts the fact that we have an optimal schedule.\qedhere
	\end{enumerate}
\end{proof}

It is worth adding that the theorem above does not generalize to problems of the same setting with three or more distinct weights. In particular, given our assumptions in \eqref{eq:cont-wsrpt}, our interchange argument relies on job $j$ being the job with the largest weight-to-remaining-work ratio among all jobs completing in the interval $[t,C_j)$. We have shown with \eqref{eq:cont-wsrpt-add} that this condition always holds when there are two distinct weights. When there are three or more distinct weights, we can easily construct examples for which this condition no longer holds, for example, by scheduling the arrival of a job with very large weight in $[t,C_j)$.

The WSRPT rule combines two well-known scheduling results: the WSPT rule (Theorem \ref{thm:smith-wspt}) and the SRPT rule (Theorem \ref{thm:schrage-srpt}). The WSRPT rule itself is not new; it has been used in other works as a popular heuristic (see \cite{batsyngps13,xiongc12}, for example). Nevertheless, to our knowledge, Theorem \ref{thm:wsrpt} is the first result on WSRPT optimality, and $1|r_j,p_j=1,pmtn|\sum \{\omega_0,\omega_1\}C_j$ is the first scheduling problem for which WSRPT is shown to be optimal.

\paragraph{Competitive Analysis}

We now consider the $\beta$-threshold rule when jobs arrive into the system over time. At each decision point, we make decisions based on an updated set of unopened jobs that accounts for any new job arrivals since our last decision point. These newly added jobs enter the queue according to their predicted priorities. Whereas our hybrid policy previously allowed a one-time switch from a preemptive policy to a nonpreemptive policy, the arrival of a high priority job could trigger preemptions when necessary, resulting in alternating preemptive and nonpreemptive regimes.

Online job arrivals add yet another layer of randomness and complexity to our model. Our efforts in competitive analysis incorporating both job arrivals and imperfect predictions were not yet fruitful. In what follows, we present our results when job arrivals are present with perfect type predictions.

Let $\mathsf{OPT}$ denote the offline optimum obtained by WSRPT, and let $\mathsf{ALG_0}$ denote the performance of the online $\beta$-threshold policy when job priorities are known a priori. The main difference between these two policies under consideration is that we are able to preempt a job whenever necessary in $\mathsf{OPT}$, but may do so at most once at a job's $\alpha$-point in $\mathsf{ALG}$.

The online policy $\mathsf{ALG_0}$ assumes that true job priorities are given to us at time of job arrival. This is a deterministic online problem where preemptions are limited to $\alpha$-points, so we might express this problem as $1|r_j,p_j=1,\alpha\text{-}pmtn|\sum \{\omega_0,\omega_1\}C_j$ in the scheduling notation of Graham et al. \cite{grahamllr79}. We follow the $\beta$-threshold rule at each decision point, where our set of unopened jobs includes jobs that have arrived since our last decision point. We review each decision in detail to highlight that each of our decisions are consistent with WSRPT. First, the existence of any unprocessed type 0 job will trigger a preemption at an $\alpha$-point. Preempting a type 1 job at an $\alpha$-point is WSRPT-consistent since, by Assumption \ref{asm:wsrpt}, $$\omega_1<\omega_0(1-\alpha)\iff \frac{\omega_1}{1-\alpha}<\frac{\omega_0}{1}.$$ Type 0 jobs will then complete nonpreemptively. When only type 1 jobs remain, any partially processed type 1 job will be processed to completion before we move on to an unopened type 1 job. This is consistent with the $\beta$-threshold rule, the SRPT rule, and by extension, the WSRPT rule. Thus, when true job types are known a priori, $\mathsf{ALG_0}$ is an optimal policy for $1|r_j,p_j=1,\alpha\text{-}pmtn|\sum \{\omega_0,\omega_1\}C_j$. We now compare its performance against $\mathsf{OPT}$. Our proofs frequently rely on the following inequality, widely known as the mediant inequality.

\paragraph{The Mediant Inequality.} For any positive real numbers $a,b,c,d>0$, $$\frac{a+b}{c+d}\leq\max\left(\frac{a}{c},\frac{b}{d}\right).$$


\begin{theorem}
	$\mathsf{ALG_0}$ is $\max\left(1+\alpha,\frac{2}{1+\alpha}\right)$-competitive, and $\sqrt{2}$-competitive if we choose $\alpha=\sqrt{2}-1$. \label{thm:alg0}
\end{theorem}
\begin{proof}
	We proceed by running the online $\beta$-threshold policy and WSRPT in parallel. Both policies schedule the same set of $n$ jobs arriving over time, where true job priorities are immediately observable upon job arrival. We refer to the schedule generated by WSRPT as the optimal schedule.
	
	We first discuss some reasonable assumptions we can impose on the data. Without loss of generality, we assume $\min_j r_j=0$, and that there is at least one job of each type in the dataset. We also assume that each of these $n$ jobs are processed without idle time in the optimal schedule so that the last job completes at time $n$. To see why, first observe that both policies are work-conserving. Any dataset that prompts a machine to become idle in an optimal schedule will simultaneously create idle time in a schedule generated by our online policy. Let us partition the dataset whenever there is idle time. Because our objective functions are linear, we can apply the mediant inequality to the competitive ratio based on said partition. Thus it suffices to consider a set of jobs that does not generate idle time.
	
	We proceed by identifying ways to further partition our set of jobs until we have a minimal set of jobs that gives the worst case performance. In order to do so, we need the following claims.
	
	\begin{claim}
		Type 1 jobs begin processing at the same time in $\mathsf{OPT}$ and in $\mathsf{ALG_0}$. This start time is always integer.\label{claim:type1-start}
	\end{claim}
	\begin{proof}[Proof of Claim \ref{claim:type1-start}]
		Let $t>0$ be any time at which a type 1 job begins its processing in an optimal schedule. By the optimality of the WSRPT rule, every type 0 job released prior to $t$ has completed by $t$, and no type 1 job that has begun its processing prior to $t$ is left unfinished. Integrality of $t$ follows naturally.
		
		Type 1 jobs also start at integer time points in the schedule generated by our online $\beta$-threshold policy because the policy requires that any partially processed jobs be completed before opening a new type 1 job. Consider time $t$ as defined above. The optimal schedule implies that all type 0 jobs released prior to $t$ have release dates no later than $t-1$. Because there is at least one decision point in the interval $[t-1,t)$, any type 0 job released prior to $t$ must have completed by $t$ in the schedule generated by the online policy. Then, by our assumption that precludes any idle time in the schedule, integer units of work have been done on type 1 jobs by $t$ and the result follows.
	\end{proof}
	
	Claim \ref{claim:type1-start} allows us to partition the schedule whenever a type 1 job begins its processing. Within each partitioned block, the same set of jobs will have completed processing in both policies. Thus, by the mediant inequality, we consider one such block. This implies that it suffices to consider a dataset with exactly one type 1 job. We call this job $k$. Without loss of generality, we assume that job $k$ begins its processing at time 0. Let $C_k^*$ and $ C_k^0$ denote the completion time of job $k$ in $\mathsf{OPT}$ and $\mathsf{ALG_0}$, respectively.
	
	\begin{claim}
		$C_k^*\geq C_k^0$.\label{claim:type1-cj}
	\end{claim}
	\begin{proof}[Proof of Claim \ref{claim:type1-cj}]
		First observe that $C_k^*$ and $C_k^0$ are positive integers for the same reasons given in the proof of Claim \ref{claim:type1-start}. Suppose on the contrary that $C_k^* < C_k^0$. Then there exists some job $\ell\neq k$ of type 0 that is \textit{not} processed prior to $C_k^*$ in the optimal schedule that is being processed in $\mathsf{ALG_0}$ at time $C_k^*$. Preemptions only occur at $\alpha$-points in $\mathsf{ALG_0}$, so job $\ell$ must have begun its processing at $C_k^*-(1-\alpha)$, which implies that $r_\ell\leq C_k^*-(1-\alpha)$. By Assumption \ref{asm:wsrpt}, it follows that $r_\ell\leq C_k^*-(1-\alpha) < C_k^*-(\omega_1/\omega_0)$.
		
		The optimal schedule follows WSRPT, so delaying the processing of any type 0 job in favor of completing job $k$ would occur only if a type 0 job arrives at such a time that the remaining work in job $k$, $x_k$, satisfies $\omega_1/x_k > \omega_0/1 \iff x_k<\omega_1/\omega_0$. Said differently, only those type 0 jobs arriving after time $C_k^*-(\omega_1/\omega_0)$ would be processed outside of the $\left[0,C_k^*\right)$ interval in an optimal schedule, and by Assumption \ref{asm:wsrpt}, also outside of $\left[0,C_k^*\right)$ in $\mathsf{ALG_0}$. Since $r_\ell < C_k^*-(\omega_1/\omega_0)$, job $\ell$ should have completed before job $k$ in an optimal schedule, which establishes the desired contradiction.
	\end{proof}
	
	An important byproduct of the proof of Claim \ref{claim:type1-cj} is that every job that completes in the interval $\left[0,C_k^*\right)$ in $\mathsf{ALG_0}$ also completes within the same interval in an optimal schedule. Both policies are work-conserving, so the converse also holds. Let $\mathcal{A}$ denote the set of jobs completing in this interval. Then, by another application of the mediant inequality, it suffices to consider the set of jobs $\mathcal{A}$, where job $k$ is the only type 1 job therein. An immediate consequence of this is an upper bound of $\max(\alpha,1-\alpha)$ on the delay in type 0 job completion times in $\mathsf{ALG_0}$ relative to those in $\mathsf{OPT}$. Intuitively, this bound captures how long a type 0 job will have to wait until the next decision point while job $k$ is being processed. Thus,
	\begin{align*}
		\frac{\mathsf{ALG_0}}{\mathsf{OPT}} &= \frac{\omega_1 C_k^0+\sum_{j\in\mathcal{A}\setminus\{k\}}\omega_0 C_j^0}{\omega_1 C_k^*+\sum_{j\in\mathcal{A}\setminus\{k\}}\omega_0 C_j^*}\\
		&\leq \frac{\sum_{j\in\mathcal{A}\setminus\{k\}}\cancel{\omega_0} C_j^0}{\sum_{j\in\mathcal{A}\setminus\{k\}}\cancel{\omega_0} C_j^*}&&\text{by the mediant inequality, since $C_k^*\geq C_k^0$}\\
		&\leq\max_{j\in\mathcal{A}\setminus\{k\}}\left(\frac{C_j^0}{C_j^*}\right)&&\text{by the mediant inequality}\\
		&= \max_{j\in\mathcal{A}\setminus\{k\}}\left(\frac{C_j^*+\delta_j}{C_j^*}\right)
	\end{align*} 
	where $\delta_j$ is job $j$'s delay in completion in $\mathsf{ALG_0}$ relative to $\mathsf{OPT}$. Finally, using $\delta_j\leq\max(\alpha,1-\alpha)$ and finding the earliest possible completion time in the optimal schedule for each type of delay,
	\begin{align*}
		\max_{j\in\mathcal{A}\setminus\{k\}}\left(\frac{C_j^*+\delta_j}{C_j^*}\right)&\leq\max\left(1+\frac{\alpha}{1},1+\frac{1-\alpha}{1+\alpha}\right)\\
		&=\max\left(1+\alpha,\frac{2}{1+\alpha}\right),
	\end{align*}
	which proves the result.
\end{proof}

Our results offer some guidance as to which values of $\alpha$ might be effective when dealing with limited preemption points. But given job arrivals and perfect predictions, the $\beta$-threshold rule is 2-competitive regardless of the $\alpha$ value that we choose. Given imperfect predictions without job arrivals, on the other hand, we obtain our earlier result given in Theorem \ref{thm:betaCR}. Competitive analysis featuring both uncertainties remains an open problem.\\

\begin{figure}[h]
	\centering
	\includegraphics[scale =.4]{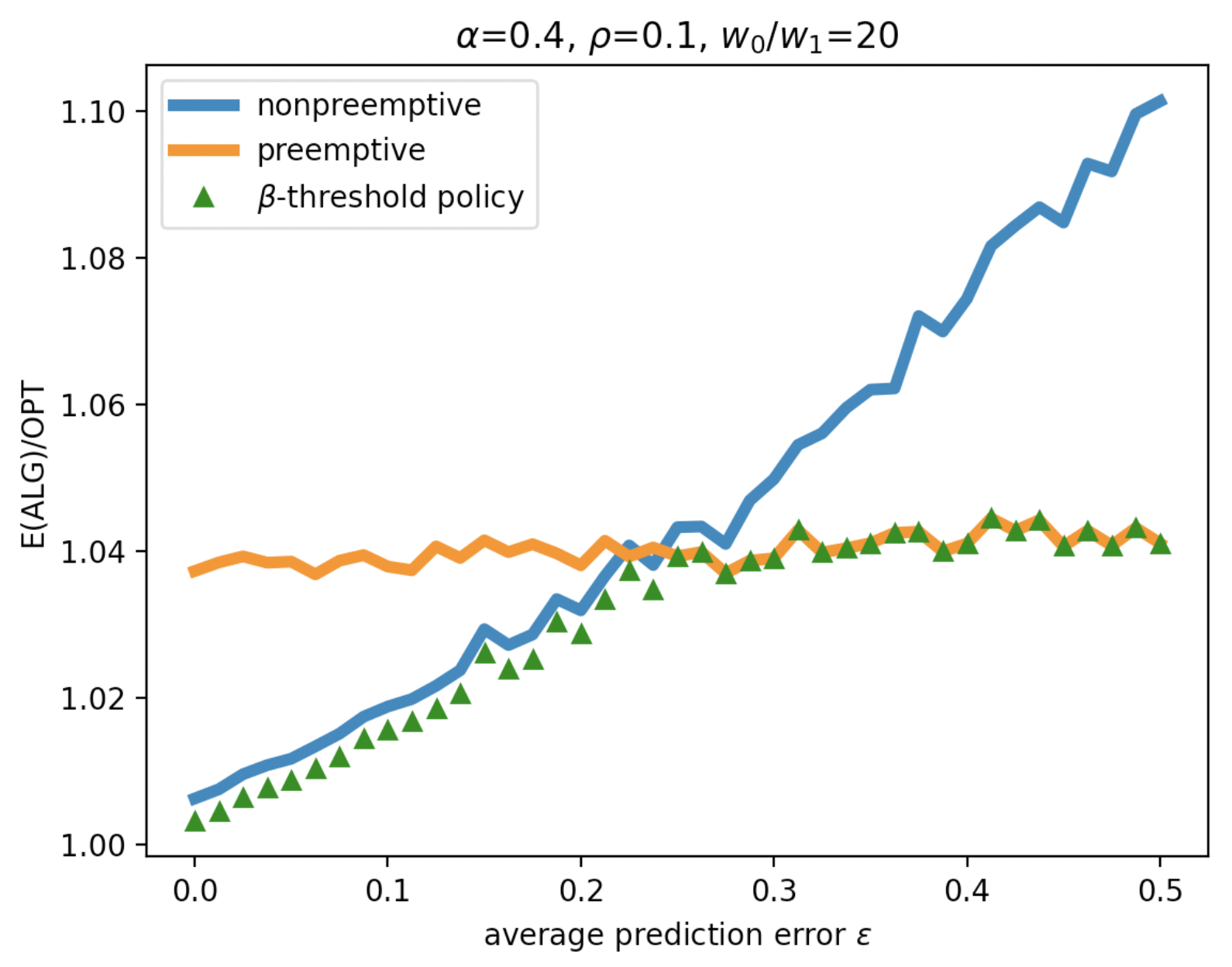}
	\caption{Expected performance when jobs arrive over time}
	\label{fig:release_dates}
\end{figure}

Despite the lack of theoretical guarantees, empirical evaluations of the $\beta$-threshold rule under realistic job arrival scenarios and imperfect prediction show that our policy still performs very well. Figure \ref{fig:release_dates} plots the expected performance of the $\beta$-threshold rule as a function of prediction error, where performance is normalized by the offline optimum obtained by WSRPT. For illustrative purposes, we assume $\varepsilon_0=\varepsilon_1$. In this plot, we assume that jobs are arriving according to a Poisson arrival process with mean interarrival time 0.9 (given unit processing times). The figure shows that our policies exhibit near-optimal performance, and that our $\beta$-threshold rule of alternating between the nonpreemptive and preemptive regimes outperforms both non-adaptive policies when we are given high quality advice.

Our experiments thus far reveal that our original stylized model without job arrivals results in the worst-case performance. This is surprising to us, and also somewhat counterintuitive given classic results in scheduling theory involving job release dates. One possible explanation for this could be the higher opportunity costs of misprediction stemming from having a long line of jobs waiting in the queue, but we do not have a good answer for this yet.

\section{Discussion and Future Directions}\label{sectionch2:discussion}

The work presented in this paper was motivated by recent interest in using machine learning algorithms for patient triage and prioritization. We modeled this as a learning-augmented online scheduling problem in which we are given good but imperfect predictions of patient risk, and sought to capture the trade-off between the need to prioritize emergency cases and the potential costs of misprediction. We presented a simple threshold-based policy that addressed these concerns and proved that our policy is in fact the best possible in certain stylized settings. The policy was also shown to remain effective in more realistic settings.

The model that we studied is grounded in reality. For many radiologists, preemptions and interruptions are simply facts of life, as is the fact that they are trained to collect information in real time while processing each patient case. In that sense, our policy recommendation is intuitive and easy to implement, and more importantly, does not require an overhaul of existing systems and Modality Worklists that are already in place. While it would be impossible to implement the $\beta$-threshold rule by the book in a clinical setting, we do believe that our policy can offer qualitative guidance on how to think about and respond to predictions of patient risk in connection with other input parameters.

That said, our work in this area is far from complete. Several concrete next steps have been outlined in Section \ref{sectionch2:extensions}, including exact characterizations of performance with probabilistic classifiers or with probabilistic learning outcomes. Theoretical guarantees of performance of the $\beta$-threshold rule with online job arrivals also remain an open problem.

Even beyond these extensions, there are many interesting directions that we can explore for future research. One natural direction would be to generalize our stylized model by allowing granularity in prioritization schemes beyond a binary classification of urgent vs. non-urgent. From a practical perspective, clinics tend to have their own internal methods of categorizing urgency levels. For example, the Department of Radiology at the Weill Cornell Medical Center categorizes urgency levels by the following:
\begin{itemize}
	\item Critical (JCAHO\footnote{Joint Commission on Accreditation of Healthcare Organizations}-designated): immediate communication required
	\item Emergent: immediate communication required
	\item Urgent: communication required in under 4 hours
	\item Important: closed-loop communication required but not in an urgent time frame (1-2 week limit).
\end{itemize}
While this is clearly a natural next step to consider, it is less evident whether our optimal policy structure extends under this more general setting. We have observed, for example, that the optimality of the WSRPT rule breaks immediately upon adding a third priority class.

Another direction would be to consider various preemptive strategies that better reflect clinical scenarios. In our scheduling formulation, preemptions could be used to model the many different ways in which radiologists learn true job types over time. One extension might be to consider multiple $\alpha$-points of preemption. For example, given $0<\alpha_1< \alpha_2<\dots <1$, we might imagine radiologists having improved confidence about a job's true priority with additional time spent processing that job. We could also consider varying preemption points for each job, for instance by letting job $j$ preempt at a unique $\alpha_j$-point once a radiologist meets a certain level of confidence. It would then be interesting to observe how performance evolves as a function of these preemption confidence levels.

In a similar vein, it is often the case that preemption comes at a cost. Our model assumes that the work required to complete an interrupted job is exactly the same as if it had not been interrupted. Realistically, it might take a while for a radiologist to warm up to a job, in which case restarting a previously preempted job may require an extra factor of $\gamma>1$ in processing time. Early attempts at tackling this problem with friction costs have not been successful due to difficulties in having to differentiate decision points by continuity in job processing. 

Continued advances in machine learning techniques mean that, over time, algorithms will likely become better at detecting abnormalities in medical images. Our current model assumes fixed error rates based on guarantees on expected generalization error, but we could also consider applying Bayesian inference techniques to update error rates over time based on observed data. This might lead to an adaptive $\beta$-threshold policy for which we might seek convergence results.

Finally, in the spirit of scheduling research, we could consider how the policy performs when there are multiple radiologists, i.e., parallel machines.

\bibliography{library}
\bibliographystyle{plainnat}

\end{document}